\documentclass{lmcs} %
\pdfoutput=1

\usepackage{lastpage}
\lmcsdoi{22}{2}{11}
\lmcsheading{}{\pageref{LastPage}}{}{}%
{Oct.~16,~2023}{Apr.~29,~2026}{}

\usepackage[utf8]{inputenc}
\usepackage[english]{babel}

\usepackage{amsfonts}
\usepackage{amssymb}
\usepackage{stmaryrd} %

\usepackage{subcaption}

\usepackage{soul} %

\usepackage{multirow}
\usepackage[ruled,vlined,linesnumbered]{algorithm2e}

\definecolor{USPNcobalt}{HTML}{293358}
\definecolor{USPNocre}{HTML}{8b7d6d}
\definecolor{USPNblanc}{HTML}{ffffff}
\definecolor{USPNceruleen}{HTML}{354878}
\definecolor{USPNsable}{HTML}{ad947e}

\usepackage[
		pdfauthor={Étienne André, Didier Lime, Olivier H. Roux},%
		pdftitle={Dense Integer-Complete Synthesis for Bounded Parametric Timed Automata},
		breaklinks  = true,
		colorlinks  = true,
		citecolor   = USPNsable,
		linkcolor   = USPNocre,
		urlcolor    = USPNceruleen,
	]{hyperref}

\usepackage[capitalise,english,nameinlink]{cleveref} %
\crefname{line}{\text{line}}{\text{lines}} %
\crefname{lemC}{Lemma}{Lemmas}
\crefname{propertyC}{Property}{Properties}

\usepackage{booktabs}

\usepackage{siunitx}

\usepackage{tikz}
\usetikzlibrary{automata, arrows, positioning, shadings}

\tikzstyle{PTA}=[auto, ->, >=stealth']
\tikzstyle{location}=[minimum size=12pt, circle, fill=blue!20, draw=black, text=black, inner sep=1.5pt, initial text={}] %
\tikzstyle{invariant}=[yshift=3, rectangle, draw=black, fill=white, text=black, inner sep=1pt]

\tikzstyle{decidable}=[green!80!black,-]
\tikzstyle{undecidable}=[red!80!black, dashed,-]
\tikzstyle{contribution}=[very thick]
\tikzstyle{legende}=[draw=none, font=\scriptsize]
\tikzstyle{unknown}=[black,dotted,thick]

\definecolor{coloract}{rgb}{0.50, 0.70, 0.30}
\definecolor{colorclock}{rgb}{0.4, 0.4, 1}
\definecolor{colordisc}{rgb}{1, 0, 1}
\definecolor{colorloc}{rgb}{0.4, 0.4, 0.65}
\definecolor{colorparam}{rgb}{.6, 0.2, 0.0}
\newcommand{\styleact}[1]{\ensuremath{\textcolor{coloract}{#1}}}
\newcommand{\styleclock}[1]{\ensuremath{\textcolor{colorclock}{#1}}}

\newcommand{\styleparam}[1]{\ensuremath{\textcolor{colorparam}{#1}}}

\newcommand{\init}{_0}

\newcommand{\A}{\ensuremath{\mathcal{A}}}
\newcommand{\Actions}{\Sigma}
\newcommand{\action}{a}
\newcommand{\assign}{\leftarrow}

\newcommand{\C}{\ensuremath{\mathbf{C}}}

\newcommand{\Constr}{\ensuremath{\mathsf{CS}}}
\newcommand{\Conv}{\ensuremath{\mathsf{Conv}}}
\newcommand{\Clock}{X} %
\newcommand{\ClockCard}{H} %
\newcommand{\clock}{x} %
\newcommand{\clocki}[1]{\styleclock{\clock_{#1}}}
\newcommand{\clockitext}[1]{\ensuremath{\clock_{#1}}}
\newcommand{\clockx}{\styleclock{\clock}} %
\newcommand{\clocky}{\styleclock{y}} %
\newcommand{\clockxtext}{\ensuremath{\clock}} %
\newcommand{\clockval}{w} %
\newcommand{\edge}{e}
\newcommand{\Edges}{E}
\newcommand{\fleche}[1]{\stackrel{#1}{\rightarrow}}
\newcommand{\TransConcrete}{\ensuremath{{\rightarrow}}}
\newcommand{\FlecheConcrete}[1]{\stackrel{#1}{\rightarrowtail}}
\newcommand{\TransConcreteEdge}{\ensuremath{{\rightarrowtail}}}
\newcommand{\FlecheSymbolic}[1]{\stackrel{#1}{\Rightarrow}}
\newcommand{\TransSymb}{\Rightarrow} %

\newcommand{\setN}{{\mathbb N}}
\newcommand{\setQ}{{\mathbb Q}}
\newcommand{\setQplus}{\setQ_{+}} %
\newcommand{\setR}{{\mathbb R}}
\newcommand{\setRplus}{\setR_{+}} %
\newcommand{\setZ}{{\mathbb Z}}
\newcommand{\subR}{{\mathbb D}}
\newcommand{\guard}{g}

\newcommand{\IH}{\ensuremath{\mathsf{IH}}}

\newcommand{\Init}{\mathbf{Init}}

\newcommand{\IV}{\ensuremath{\mathsf{IV}}}

\newcommand{\RIEF}{\ensuremath{\mathsf{RIEF}}}
\newcommand{\RIAF}{\ensuremath{\mathsf{RIAF}}}
\newcommand{\IEF}{\ensuremath{\mathsf{IEF}}}
\newcommand{\AF}{\ensuremath{\mathsf{AF}}}
\newcommand{\EF}{\ensuremath{\mathsf{EF}}}
\newcommand{\IAF}{\ensuremath{\mathsf{IAF}}}

\newcommand{\invariant}{I}
\newcommand{\K}{\styleSymbStatesSet{K}}
\newcommand{\KTrue}{\ensuremath{{\setQplus}^{\Param}}}

\newcommand{\KBlock}{\ensuremath{\K_\mathit{block}}}
\newcommand{\KGood}{\ensuremath{\K_\mathit{good}}}
\newcommand{\KLive}{\ensuremath{\K_\mathit{live}}}
\newcommand{\Kgood}{\ensuremath{\K_\mathit{good}}}
\newcommand{\LargestC}{{\ensuremath{\textcolor{colorok}{M}}}} %
\newcommand{\loc}{\ensuremath{\ell}} %
\newcommand{\locinit}{\loc\init}
\newcommand{\Loc}{L} %
\newcommand{\lterm}{\mathit{lt}}
\newcommand{\Param}{P} %
\newcommand{\param}{p} %
\newcommand{\ParamDomain}{\ensuremath{\mathit{PD}}} %
\newcommand{\paramp}{\styleparam{\param}}
\newcommand{\parama}{\styleparam{a}}
\newcommand{\paramb}{\styleparam{b}}
\newcommand{\parami}[1]{\styleparam{\param_{#1}}}
\newcommand{\paramq}{\ensuremath{\styleparam{q}}}
\newcommand{\paramptext}{\ensuremath{\param}}
\newcommand{\paramqtext}{\ensuremath{q}}
\newcommand{\paramatext}{\ensuremath{a}}
\newcommand{\parambtext}{\ensuremath{b}}
\newcommand{\paramctext}{\ensuremath{c}}
\newcommand{\paramitext}[1]{\ensuremath{\param_{#1}}}
\newcommand{\paramLambda}{\ensuremath{\lambda}}
\newcommand{\paramSigma}{\ensuremath{\sigma}}
\newcommand{\paramTS}{\ensuremath{\mathit{TS}}}
\newcommand{\paramcard}{M} %
\newcommand{\pval}{v} %
\newcommand{\plterm}{\mathit{plt}}
\newcommand{\Q}{Q}
\newcommand{\qinit}{q\init}
\newcommand{\resets}{R}

\newcommand{\LocsTarget}{\ensuremath{G}} %
\newcommand{\state}{s} %
\newcommand{\Succ}{\mathbf{Succ}}
\newcommand{\symtree}{\ensuremath{T^\infty}}
\newcommand{\timelapse}[1]{#1^\nearrow}
\newcommand{\timepast}[1]{#1^\swarrow}
\newcommand{\TPS}{\ensuremath{\mathsf{TP}}}
\newcommand{\RITPS}{\ensuremath{\mathsf{RITP}}}
\newcommand{\Traces}{\mathit{Traces}}
\newcommand{\treeprefix}{\ensuremath{T}}
\newcommand{\Variables}{\ensuremath{V}}
\newcommand{\variable}{\ensuremath{\nu}}
\newcommand{\varval}{\ensuremath{\eta}}
\newcommand{\varrun}{\ensuremath{\rho}} %
\newcommand{\wv}[2]{#1|#2} %

\newcommand{\valuations}[1]{\ensuremath{\llbracket{#1}\rrbracket}}

\newcommand{\projectP}[1]{\ensuremath{#1{\downarrow_{\Param}}}}
\newcommand{\reset}[2]{\ensuremath{[#1]_{#2}}}
\newcommand{\valuate}[2]{\ensuremath{#2(#1)}}

\newcommand{\cardinality}[1]{\ensuremath{\lvert #1 \rvert}}

\newcommand{\PZG}{\ensuremath{\styleSymbStatesSet{PSG}}} %
\newcommand{\styleSymbStatesSet}[1]{\ensuremath{\mathbf{#1}}}
\newcommand{\Passed}{\styleSymbStatesSet{Passed}}
\newcommand{\symbstate}{\ensuremath{\styleSymbStatesSet{s}}} %
\newcommand{\SymbStates}{\ensuremath{\styleSymbStatesSet{S}}} %
\newcommand{\symbstateinit}{\symbstate\init} %

\newcommand{\TimeOut}{\textbf{\textcolor{red}{T.O.}}}

\newcommand{\imitator}{\textsf{IMITATOR}}
\newcommand{\romeo}{\textsc{Roméo}}

\newcommand{\cylinder}[1]{\ensuremath{\textsf{Cyl}_{#1}}}
\newcommand{\Ext}[2]{\ensuremath{\textsf{Ext}^{#1}_{#2}}}

\usepackage{amsthm}
\theoremstyle{plain}

\newtheorem{lemma}[thm]{Lemma}

\newtheorem{proposition}[thm]{Proposition}
\newtheorem{theorem}[thm]{Theorem}

\theoremstyle{definition}
\newtheorem{definition}[thm]{Definition}
\newtheorem{example}[thm]{Example}

\theoremstyle{remark}
\newtheorem{remark}[thm]{Remark}

\theoremstyle{thmC}
\newtheorem{propertyC}[thm]{Property}
\newcommand{\defProblem}[3]
{%
	\noindent\fcolorbox{black}{USPNsable!20}{
	\begin{minipage}{.95\columnwidth}
		\textbf{#1:}\\
		\textsc{Input}: #2\\
		\textsc{Problem}: #3
	\end{minipage}
}

	\smallskip

}

	\definecolor{colorok}{RGB}{0,0,0}

\newcommand{\eg}{\textcolor{colorok}{\textit{e.g.,}}\xspace}
\newcommand{\ie}{\textcolor{colorok}{\textit{i.e.,}}\xspace}
\newcommand{\wrt}{\textcolor{colorok}{w.r.t.}\ }

\begin{document}
\sloppy

\title[Dense Integer-Complete Synthesis for Parametric Timed Automata]{Dense Integer-Complete Synthesis for Bounded Parametric Timed Automata}

\thanks{%
	This is an extended version of the paper by the same authors published in the proceedings of the 9th International Workshop on Reachability Problems ({RP} 2015)~\cite{ALR15}.
	This work has been supported by
		the ANR-NRF French-Singaporean research program ProMiS (ANR-19-CE25-0015 / 2019 ANR NRF 0092)
		and ANR BisoUS (ANR-22-CE48-0012).}
\author[É.~André]{Étienne André\lmcsorcid{0000-0001-8473-9555}}[a,b]
\author[D.~Lime]{Didier Lime\lmcsorcid{0000-0001-9429-7586}}[c]
\author[O.~H.~Roux]{Olivier~H.~Roux\lmcsorcid{0000-0003-1665-0481}}[c]

\address{Université Sorbonne Paris Nord, LIPN, CNRS UMR 7030, F-93430 Villetaneuse, France}	%
\address{Institut Universitaire de France (IUF)}

\address{Nantes Université, École Centrale Nantes, CNRS, LS2N, UMR 6004, F-44000 Nantes, France}	%

\begin{abstract}
	Ensuring the correctness of critical real-time systems, involving concurrent behaviours and timing requirements, is crucial.
	Timed automata extend finite-state automata with clocks, compared in guards and invariants with integer constants.
	Parametric timed automata (PTAs) extend timed automata with timing parameters.
	Parameter synthesis aims at computing dense sets of valuations for the timing parameters, guaranteeing a good behaviour.
	However, in most cases, the emptiness problem for reachability (\ie{} the emptiness of the parameter valuations set for which some location is reachable) is undecidable for~PTAs and, as a consequence, synthesis procedures do not terminate in general, even for bounded parameters.
	In this paper, we introduce a parametric extrapolation, that allows us to derive an underapproximation in the form of symbolic sets of valuations containing not only all the integer points ensuring reachability, but also all the (non-necessarily integer) convex combinations of these integer points, for general PTAs with a bounded parameter domain.
	We also propose two further algorithms synthesizing parameter valuations guaranteeing unavoidability, and preservation of the untimed behaviour \wrt{}a reference parameter valuation, respectively.
	Our algorithms terminate and can output sets of valuations arbitrarily close to the complete result.
	We demonstrate their applicability and efficiency using the tools \romeo{} and \imitator{} on several benchmarks.
\end{abstract}

\maketitle

\section{Introduction}
\label{section:intro}

The verification of software or hardware systems mixing time and concurrency is a notoriously difficult problem.
Timed automata (TAs)~\cite{AD94} are a powerful formalism for which many interesting problems (including the reachability of a location) are decidable.
However, the classical definition of TAs is not tailored to verify systems only partially specified, especially when the value of some timing constants is not yet known.
Parametric timed automata (PTAs)~\cite{AHV93} overcome this problem by allowing the specification and the verification of systems where some of the timing constants are parametric.
This expressive power comes at the price of the undecidability of most interesting problems (see~\cite{Andre19STTT} for a survey).

Synthesis for PTAs aims at deriving set of valuations for which a given property (such as reachability or unavoidability) holds.
Ideally, this set of parameter valuations comes in a \emph{symbolic} form, \eg{} using a finite union of polyhedra representing \emph{dense} sets of valuations.
This density can be useful notably in the context of robustness or implementability of a system (see, \eg{}~\cite{BMS13}).
The goal of this work is to achieve synthesis of \emph{dense} sets of parameter valuations, while guaranteeing the termination of our algorithms.\label{newtext:intro:dense}

\subsection{Related work}
The simple problem of the emptiness of the valuations set such that some location of a PTA is reachable (``reachability-emptiness'') is undecidable in both discrete and dense-time settings~\cite{AHV93}, even with only three parametric clocks (\ie{} clocks compared to a parameter)~\cite{Miller00,BBLS15},
and even with only strict constraints~\cite{Doyen07}.
It is decidable for a single parametric clock in discrete and dense time~\cite{AHV93}, and in discrete time with two parametric clocks and one parameter (and arbitrarily many non-parametric clocks)~\cite{BO14,GH21}, or for a single parametric clock (with arbitrarily many non-parametric clocks)~\cite{BBLS15}.
More complex properties expressed in parametric TCTL have been studied in~\cite{Wang96,BR07,BDR08}.

Subclasses of PTAs have been studied, most notably L/U-PTAs, in which the parameters are partitioned into parameters that can only be compared to a clock as an upper bound, or as a lower bound.
The reachability-emptiness problem is decidable for L/U-PTAs~\cite{HRSV02}, but no synthesis algorithm is provided, and there are indeed practical difficulties in proposing one~\cite{JLR15}.
When further restricting to L- or U-PTAs (with only lower-bound, resp.\ upper-bound parameters), the integer-valued parameters can be synthesised~\cite{BlT09}.
However, the full TCTL-emptiness problem was proved to be undecidable even for the restricted class of U-PTAs~\cite{ALR18FORMATS}.

In~\cite{ACEF09}, the inverse method takes advantage of a reference parameter valuation.
When it terminates, this algorithm outputs a (possibly underapproximated) set of parameter valuations for which the discrete behaviour (set of traces seen as alternating sequences of locations and actions) is exactly the same as the one of the reference valuation.
The traces preservation problem (\ie{} given a reference parameter valuation whether there exists another valuation for which the discrete behaviour is the same) is undecidable for both general PTAs and L/U-PTAs~\cite{ALM20}.

In~\cite{JLR15}, the focus is on integer-valued bounded parameters (but still considering dense-time), for which many problems are obviously decidable, and symbolic algorithms are provided to compute the set of correct integer parameter valuations ensuring a reachability (``\IEF{}'') or unavoidability (``\IAF{}'') property.
A drawback is that returning only integer points prevents designers to use the synthesised constraint to study the robustness or implementability of their system.
Control of real-time systems with integer-valued bounded parameters was then studied in~\cite{JLR22} using similar techniques.

Finally note that an extrapolation similar to ours~\cite{ALR15} was later proposed in~\cite{BBBC16}.

\subsection{Contribution}\label{ss:contribution}

We propose here a terminating algorithm that computes a dense under-approximation of the set of parameter valuations ensuring reachability in bounded PTAs (\ie{} PTAs with a bounded parameter domain).
We also propose two further algorithms ensuring unavoidability and preservation of the untimed behaviour \wrt{}a reference parameter valuation, respectively.
These under-approximations are ``integer-complete'' in the sense that they are guaranteed to contain at least all the correct integer valuations given in the form of a finite union of polyhedra;
this is to say, only some non-integer (rational) points may be missing.\label{newtext:almostcomplete}
To the best of our knowledge, these algorithms are the first synthesis algorithms that return a dense integer-complete result for a subclass of PTAs (namely bounded PTAs) for which the corresponding emptiness problems are undecidable;
in fact, we are not aware of any terminating synthesis algorithm with a dense result with a partial completeness guarantee.
Note that the synthesis algorithms for two subclasses of L/U-PTAs (namely L-PTAs and U-PTAs)~\cite{BlT09} are complete, but the setting is restricted to integer valuations.

While of great practical interest, our algorithms are in essence quite similar to those of~\cite{JLR15}. We however demonstrate that while the algorithms from~\cite{JLR15} also return a symbolic representation of the ``good'' integer parameter valuations, interpreting the result of the \IAF{} algorithm as dense is not correct in the sense that some non-integer parameter valuations in that result may not ensure the unavoidability property. Furthermore, since we produce rational-valued parameter valuations, we cannot use anymore the result from~\cite{JLR15} ensuring termination of the algorithms, which allows to derive a bound on clock valuations but relies on the parameters being bounded integers.
One of the main \emph{technical} contributions of this paper is therefore the derivation of a maximum-constant-based parametric extrapolation operator for bounded PTAs that ensures termination of our algorithms. To the best of our knowledge this operator is the first of its kind.

Finally, we have implemented the three algorithms for reachability, unavoidability and traces preservation, and we report on them.

\subsection{Comparison with~\cite{ALR15}}
This manuscript is a significant extension of~\cite{ALR15}.
In addition to the inclusion of all proofs (missing in~\cite{ALR15}) and of additional examples, we rephrased most definitions and explanations; we also added the new section on traces preservation (\cref{ss:IM}).
We also report on experiments in the new \cref{section:implementation-romeo,section:implementation-imitator}.

\subsection{Outline}
We first recall the necessary definitions in \cref{section:preliminaries}.
We present our parametric extrapolation in \cref{section:extrapolation}.
We then introduce our terminating algorithms (namely \RIEF{}, \RIAF{}, \RITPS{}) in \cref{section:algos}.
We report on our implementations, and we relate experiments on several benchmarks using \romeo{} (to evaluate \RIEF{} and \RIAF{}) in \cref{section:implementation-romeo}, and using \imitator{} (to evaluate \RITPS{}) in \cref{section:implementation-imitator}.
We conclude in \cref{section:conclusion}.

\section{Preliminaries}
\label{section:preliminaries}
\subsection{Clocks, parameters and constraints}\label{ss:constraints}

Let $\setN$,
$\setZ$,
$\setQplus$,
$\setRplus$,
$\setR$
denote the sets of non-negative integers, integers, non-negative rationals, non-negative reals and reals respectively.

Let $\subR \subseteq \setR$.
Given a finite variables set~$\Variables$, a $\subR$-valuation over~$\Variables$ is a function from~$\Variables$ to~$\subR$.

\subsubsection{Constraints and convex polyhedra}\label{sss:constraints}
In the following, we assume ${\bowtie} \in \{<, \leq, \geq, >\}$.

A \emph{constraint}~$\Constr$ over a set of variables~$\Variables$ is a conjunction of inequalities of the form $\lterm \bowtie 0$, where
$\lterm$ denotes a linear term over $\Variables$ of the form $\sum_{1 \leq i \leq |\Variables|} \alpha_i \variable_i + d$, with $\variable_i \in \Variables$ and $\alpha_i, d \in \setZ$.

A valuation $\varval$ over $\Variables$ is a \emph{solution} of a constraint $\Constr$ if by replacing each occurrence of each variable $\variable$ in $\Constr$ by $\valuate{\variable}{\varval}$ simplifies to the ``true'' Boolean value.
Given a constraint~$\Constr$ over $\Variables$, we denote by $\valuations{\Constr}$ the set of its solutions.

A set of valuations $\C$ is a \emph{convex polyhedron\footnote{
Given an arbitrary order on~$\Variables$, a valuation over~$\Variables$ can be seen as a real vector of size~$|\Variables|$ whose coordinates are the values associated to the variables and thus our convex polyhedra indeed represent subsets of~$\setR^{|\Variables|}$ and the set of $\setR$-valuations over $\Variables$ is a vector space over $\setR$.%
}} if there exists a constraint $\Constr$ such that $\C=\valuations{\Constr}$.

$\C$ is topologically closed if it is a solution to a conjunction of constraints using only non-strict inequalities.

By the Minkowski-Weyl theorem (see, \eg{} \cite{Schrijver86}), topologically closed convex polyhedra can be equivalently defined using a set of vertices and extremal rays.

Let $\Conv(\C)$ denote the convex hull of set~$\C$ of valuations, that is $\Conv(\C)=\{\varval \mid \exists k,\exists \lambda_1,\ldots,\lambda_k\in[0,1],\exists\varval_1,\dots,\varval_k\in \C\text{ s.t. } \varval = \sum_{i=1}^k \lambda_i \varval_i \land \sum_{i=1}^k\lambda_i = 1\}$.\label{def:Conv}

Then for any convex polyhedron~$\C$, there exist $\varval_1,\ldots,\varval_j\in\C$, $r_1,\ldots, r_k$ some $\setR$-valuations, and $\mu_1,\ldots,\mu_k$ non-negative real numbers such that for all $\varval \in \C$, there exist $\varval' \in \Conv(\{\varval_1,\ldots,\varval_j\})$ and $\varval = \varval'+ \sum_{i=1}^k \mu_i r_i$. Without loss of generality, the $\varval_i$ can be uniquely chosen such that none of them belongs to the convex hull of the others, and they are then called the \emph{vertices} of~$\C$.
Similarly, the $r_i$ can be chosen, uniquely up to some multiplicative constant, such that none can be expressed as a non-negative combination of the others, and are then called the \emph{extremal rays} of~$\C$.
Informally, the latter represent the limit directions in which the polyhedron is infinite.

Since we consider also strict inequalities, the above generator description is incomplete in general and we would also need to introduce \emph{closure points} that play the role of vertices in the topological closure of the polyhedron.
Their handling would be similar to what we do with vertices in the following. The reader is directed to~\cite{BHZ05} for more details, and in particular Theorem~4.4 therein for an equivalent to the Minkowski-Weyl theorem for non-necessarily closed polyhedra.

Also note that any non-closed polyhedron can be represented by a closed polyhedron with one extra dimension~\cite{HPR94}.

In the sequel to keep the developments more readable we will assume all polyhedra are topologically closed and therefore deal only with vertices and rays.

\subsubsection{Clocks and parameters}

Throughout this paper, we assume a set~$\Clock = \{ \clock_1, \dots, \clock_\ClockCard \} $ of \emph{clocks}, \ie{} non-negative real-valued variables that evolve at the same rate.
A \emph{clock valuation} is an $\setRplus$-valuation over~$\Clock$.
We write $\vec{0}$ for the valuation such that for all $\clock\in\Clock$, $\vec{0}(x) = 0$.
Given $\resets \subseteq \Clock$, we define the \emph{reset} of a valuation~$\clockval$, denoted by $\reset{\clockval}{\resets}$, as follows: $\reset{\clockval}{\resets}(\clock) = 0$ if $\clock \in \resets$, and $\reset{\clockval}{\resets}(\clock)=\clockval(\clock)$ otherwise.
Given $d \in \setRplus$, $\clockval + d$ denotes the valuation such that $(\clockval + d)(\clock) = \clockval(\clock) + d$, for all $\clock \in \Clock$.
An \emph{integer} clock valuation in an $\setN$-valuation over~$\Clock$.

We assume a set~$\Param = \{ \param_1, \dots, \param_\paramcard \} $ of \emph{parameters}, \ie{} unknown constants.
A \emph{parameter valuation} $\pval$ is a $\setQplus$-valuation over~$\Param$.
An \emph{integer} parameter valuation is an $\setN$-valuation over~$\Param$.

\subsubsection{Simple clock constraints}\label{sss:guards}
Let $\plterm$ denote a \emph{parametric linear term}, \ie{} a linear term over~$\Param$.
A \emph{simple clock constraint}~$\guard$ is a constraint over $\Clock \cup \Param$ defined by a conjunction of inequalities of the form $\clock \bowtie \plterm$, \ie{}
	$\clock \bowtie \sum_{1 \leq i \leq |\Param|} \alpha_i \param_i + d$, with $\clock \in \Clock$, $\param_i \in \Param$, and $\alpha_i, d \in \setZ$.

\subsubsection{Sets of valuations}\label{sss:valuations}

Given a set~$\C$ of valuations over $\Clock \cup \Param$, and given a parameter valuation~$\pval$, $\valuate{\C}{\pval}$ denotes the set of valuations over~$\Clock$ obtained by replacing each parameter~$\param$ in~$\C$ with~$\pval(\param)$.

Given a parameter valuation $\pval$ and a clock valuation $\clockval$, we denote by $\wv{\clockval}{\pval}$ the valuation over $\Clock\cup\Param$ such that
for all clocks $\clock$, $\valuate{\clock}{\wv{\clockval}{\pval}}=\valuate{\clock}{\clockval}$
and
for all parameters $\param$, $\valuate{\param}{\wv{\clockval}{\pval}}=\valuate{\param}{\pval}$.

An \emph{integer point} is a valuation $\wv{\clockval}{\pval}$, where $\clockval$ is an integer clock valuation, and $\pval$ is an integer parameter valuation.

Given a set~$\C$ of valuations over $\Clock \cup \Param$,
we define the \emph{future} of~$\C$, denoted by $\timelapse{\C}$, as the set of valuations over~$\Clock \cup \Param$ obtained from~$\C$ by delaying all clocks by an arbitrary amount of time:
$\timelapse{\C} = \{ \wv{\clockval'}{\pval} \mid \exists \clockval,d : \clockval \in \valuate{\C}{\pval} \land
\clockval' = \clockval + d, %
d \in \setRplus \}$.
Dually, we define the \emph{past} of~$\C$, denoted by $\timepast{\C}$, as the set of valuations over~$\Clock \cup \Param$ obtained from~$\C$ by letting time pass backward by an arbitrary amount of time:
$\timepast{\C} = \{ \wv{\clockval'}{\pval} \mid \exists\clockval,d : \clockval \in \valuate{\C}{\pval} \land
\clockval' + d = \clockval , %
	d \in \setRplus \}$.

Given $\resets \subseteq \Clock$, we define the \emph{reset} of~$\C$, denoted by $\reset{\C}{\resets}$, as the set of valuations over~$\Clock \cup \Param$ obtained from~$\C$ by resetting the clocks in~$\resets$, and keeping the other clocks unchanged.
That is, \(\wv{\clockval'}{\pval} \in \reset{\C}{\resets} \) iff \(\exists \clockval : \Clock \to \setRplus \) s.t.\ \( \wv{\clockval}{\pval} \in \C \) and
	$\clockval' = \reset{\clockval}{\resets}$. %

We denote by $\projectP{\C}$ the projection of~$\C$ onto~$\Param$,
\ie{} the set of valuations over~$\Param$ defined as follows:
$\{ \pval \mid \exists \clockval \text{ s.t. } \wv{\clockval}{\pval} \in \C \}$.

Recall that, if $\C,\C'$ are convex polyhedra, then $\C \cap \C'$, $\timelapse{\C}$, $\timepast{\C}$, $\reset{\C}{\resets}$ and $\projectP{\C}$ are all convex polyhedra~\cite[pp.6-7]{JLR15}.

\subsection{Parametric timed automata}

Parametric timed automata (PTAs) extend timed automata with parameters within guards and invariants in place of integer constants~\cite{AHV93}.

\begin{definition}\label{def:PTA}
	A PTA
	$\A$ is a tuple \mbox{$\A = (\Actions, \Loc, \locinit, \Clock, \Param, \ParamDomain, \invariant, \Edges)$}, where:
	\begin{enumerate}%
		\item $\Actions$ is a finite set of actions,
		\item $\Loc$ is a finite set of locations,
		\item $\locinit \in \Loc$ is the initial location,
		\item $\Clock$ is a finite set of clocks,
		\item $\Param$ is a finite set of parameters,
		\item $\ParamDomain \subseteq \setQplus^{\Param}$ is the parameter domain, \ie{} the set of admissible parameter valuations,
		\item $\invariant$ is the invariant, assigning to every $\loc\in \Loc$ a simple clock constraint $\invariant(\loc)$,
		\item $\Edges$ is a set of edges  $\edge = (\loc,\guard,\action,\resets,\loc')$
		where
		$\loc,\loc'\in \Loc$ are the source and target locations, $\action \in \Actions$, $\resets\subseteq \Clock$ is a
		set of clocks to be reset to~$0$, and
		$\guard$ (the ``guard'') is a simple clock constraint.
	\end{enumerate}
\end{definition}
\begin{example}
	Consider the PTA in \cref{figure:example:reviewerACSD}.
	It features 4 locations, 2 clocks and 1~parameter.
	The initial location is~$\loc_1$.
	The transition from~$\loc_1$ to~$\loc_2$ is guarded by ``$\clockxtext = 1$'' and resets~$\clockxtext$ to~0.
	This PTA has no invariant.
\end{example}

In this paper, we consider \emph{bounded} PTAs.
A \emph{bounded parameter domain} is such that $\ParamDomain$ assigns to each parameter a minimum integer bound and a maximum integer bound.
That is, each parameter~$\param_i$ ranges in an interval $[a_i, b_i]$, with $a_i,b_i \in \setN$.
Hence, a bounded parameter domain is a hyperrectangle in $\cardinality{\Param}$ dimensions.

\begin{definition}[bounded PTA]\label{def:boundedPTA}
	A bounded PTA is a PTA the parameter domain of which is bounded.
\end{definition}
\begin{remark}
	In this paper, our bounded parameter domains are simple closed intervals with integer bounds.
	However, we noted in~\cite{ALR22} that the nature of these intervals (open or closed) can have an impact on decidability, notably for some subclasses of PTAs~\cite{BlT09}.
\end{remark}

Given a parameter valuation~$\pval$, we denote by $\valuate{\A}{\pval}$ the non-parametric structure where all occurrences of a parameter~$\param_i$ have been replaced by~$\pval(\param_i)$.
In the following, we may denote as a \emph{timed automaton (TA)} any such structure $\valuate{\A}{\pval}$, by assuming a rescaling of the constants:
	by multiplying all constants in $\valuate{\A}{\pval}$ by their least common denominator, we obtain an equivalent timed automaton (with integer constants, as in~\cite{AD94}).

\subsubsection{Concrete semantics}\label{sss:concrete-semantics}

Let us first recall the concrete semantics of~TAs.

\begin{definition}[Semantics of a TA]\label{definition:semanticsTA}
	Given a PTA $\A = (\Sigma, \Loc, \locinit, \Clock, \Param, \ParamDomain, \invariant, \Edges)$,
	and a parameter valuation~\(\pval\),
    the \emph{concrete semantics} of $\valuate{\A}{\pval}$ is given by the timed transition system $(\Q, \qinit, \TransConcrete)$, with
	\begin{itemize}
		\item 
            $\Q = \{ (\loc, \clockval) \in \Loc \times \setRplus^{\cardinality{\Clock}} \mid \wv{\clockval}{\pval} \in \valuations{\invariant(\loc)} \}$
        \item $\qinit = (\locinit, \vec{0}) $,
		\item delay transitions: 
            $\big((\loc, \clockval), d, (\loc, \clockval+d) \big) \in {\TransConcrete}$, with $d \in \setRplus$, if $\forall d' \in [0, d], (\loc, \clockval+d') \in \Q$;
		\item discrete transitions:
            $\big((\loc, \clockval), \edge, (\loc',\clockval') \big) \in {\TransConcrete}$ %
				if $(\loc, \clockval) , (\loc',\clockval') \in \Q$, there exists $\edge = (\loc,\guard,\action,\resets,\loc') \in \Edges$, $\clockval'= \reset{\clockval}{\resets}$, and $\wv{\clockval}{\pval} \in \valuations{\guard}$.
	\end{itemize}
\end{definition}

We assume that $\qinit \in \Q$, \ie{} $\wv{\vec{0}}{\pval} \in \valuations{\invariant(\locinit)}$, \ie{} the initial clock valuation satisfies the invariant of the initial location.\label{newtext:initial-invariant}

We refer to the states of the concrete semantics of a TA as its \emph{concrete states}.
As usual, given two concrete states $\state,\state'$, we write $\state\fleche{a}\state'$, for $a\in\setRplus\cup\Edges$, instead of $(\state, a,\state') \in {\TransConcrete}$.
We also further define relation $\TransConcreteEdge$ by $\big((\loc, \clockval), \edge, (\loc', \clockval')\big) \in \TransConcreteEdge$ if $\exists \clockval''\in\setRplus^{\cardinality{\Clock}},\edge\in\Edges, d\in\setRplus:  (\loc, \clockval) \fleche{d} (\loc,\clockval'') \fleche{\edge} (\loc',\clockval')$.
A concrete run~$\varrun$ of a TA is an alternating sequence of concrete states of $\Q$ and edges of the form
$\state_0 \FlecheConcrete{\edge_0} \state_1\FlecheConcrete {\edge_1} \cdots \FlecheConcrete{\edge_{m-1}} \state_m$, such that for all $i = 0, \dots, m-1$, $\edge_i \in \Edges$, and $(\state_i , \edge_i , \state_{i+1}) \in \TransConcreteEdge$.
We say that states $\state_i$, for $i = 0, \dots, m$, \emph{belong} to~$\varrun$.\label{newtext:belong}
Given a concrete state~$\state = (\loc, \clockval)$, we say that $\state$ is reachable (or that $\valuate{\A}{\pval}$ reaches~$\state$) if $\state$ belongs to a run of $\valuate{\A}{\pval}$.
By extension, a location~$\loc$ is reachable if there exists $\clockval$ such that $(\loc,\clockval)$ is reachable.
By extension again, a subset of locations $\LocsTarget \subseteq \Loc$ is reachable if there exists $\loc \in \LocsTarget$ such that $\loc$ is reachable.

A \emph{maximal} concrete run is a run that is either infinite, or that cannot be extended.

Given a concrete run $(\loc_0, \clockval_0) \FlecheConcrete{\edge_0} (\loc_1, \clockval_1) \FlecheConcrete {\edge_1} \cdots \FlecheConcrete{\edge_{m-1}} (\loc_m, \clockval_m)$, its corresponding \emph{trace} is
	$\loc_0  \FlecheConcrete{\action_0} \loc_1 \FlecheConcrete {\action_1} \cdots \FlecheConcrete{\action_{m-1}} \loc_m $,
		where $\action_i$ is the action of $\edge_i$.
The set of all traces of a TA~$\A$ is called its \emph{trace set}, denoted by $\Traces(\A)$.
That is, traces record both the locations and actions, but leave out the quantitative information such as the clock valuations and timing delays.
\subsubsection{Symbolic semantics}\label{sss:symbolic}
We now recall the symbolic semantics of PTAs~\cite{HRSV02,ACEF09,JLR15,AMPP22}.
Note that we usually use \textbf{bold} font to denote anything symbolic, \ie{} (sets of) symbolic states, and constraints.

\begin{definition}[Symbolic state]\label{definition:symbolic-state}
	Given a PTA~$\A$, a symbolic state of~$\A$ is in the form of a pair $(\loc, \C)$ where $\loc \in \Loc$ is a location of~$\A$, and $\C$ is a set of valuations.
\end{definition}

Given a parameter valuation~$\pval$, a symbolic state $\symbstate = (\loc, \C)$ is $\pval$-compatible if $\pval \in \projectP{\C}$.

Given a PTA~$\A$ and a parameter valuation~$\pval$, given a concrete state $(\loc, \clockval)$ of~$\valuate{\A}{\pval}$ and a symbolic state $(\loc', \C)$ of~$\A$,
we write $(\loc, \clockval) \in \valuate{(\loc', \C)}{\pval}$ if $\loc = \loc'$ and $\wv{\clockval}{\pval} \in \C$.

\begin{definition}[Successor of a set of valuations]\label{def:succ}
	Given a PTA $\A = (\Actions, \Loc, \locinit, \Clock, \Param, \ParamDomain, \invariant, \Edges)$,
	given an edge $\edge = (\loc,\guard,\action,\resets,\loc') \in \Edges$,
	given a set~$\C$ of valuations over~$\Clock \cup \Param$, we define
	\[\Succ(\C, \edge) = \timelapse{\big(\reset{(\C \cap \valuations{\guard})}{\resets} \cap \valuations{\invariant(\loc')} \big )} \cap \valuations{\invariant(\loc')}\text{.}\]
\end{definition}
\begin{definition}[Symbolic semantics]\label{def:PTA:symbolic}
	Given a PTA $\A = (\Actions, \Loc, \locinit, \Clock, \Param, \ParamDomain, \invariant, \Edges)$, %
	the symbolic semantics of~$\A$ is the labelled transition system called \emph{parametric symbolic state graph}
	$ \PZG = ( \Edges, \SymbStates, \symbstateinit, {\TransSymb} )$, with
	\begin{itemize}
		\item $\SymbStates = \{ (\loc, \C) \mid \C \subseteq \valuations{\invariant(\loc)} \}$,
		\item $\symbstateinit = \Big(\locinit, \timelapse{\big( \{ \wv{\clockval}{\pval} \mid \clockval = \vec{0} \land \wv{\clockval}{\pval} \in \valuations{\invariant(\loc_0)} \} \big)} \cap \valuations{\invariant(\loc_0)} \Big)$,
				and
		\item $\big((\loc, \C), \edge, (\loc', \Succ(\C, \edge))\big) \in {\TransSymb} $ if $\edge = (\loc,\guard,\action,\resets,\loc') \in \Edges$ and $\Succ(\C, \edge) \neq \emptyset$.
	\end{itemize}
\end{definition}

That is, in the parametric symbolic state graph, nodes are symbolic states, and arcs are labelled by \emph{edges} of the original PTA.
In the following, we denote symbolic states using bold font~(``$\symbstate$'').

Given $\symbstate = (\loc, \C)$ and $\edge = (\loc,\guard,\action,\resets,\loc') \in \Edges$,
we write $\Succ(\symbstate, \edge) = (\loc', \Succ(\C, \edge))$.\label{newtext:succ}

Given $\symbstate = (\loc, \C)$, as an abuse of notation, we write $\symbstate = \emptyset$ whenever $\C = \emptyset$.\label{newtext:KFalse}

Given a PTA~$\A$ with parametric symbolic state graph $ \PZG = ( \Edges, \SymbStates, \symbstateinit, {\TransSymb} )$, we let $\Init(\A)$ denote~$\symbstateinit$.

A symbolic run of a PTA is an alternating sequence of symbolic states and edges of the form 
$\symbstate_0 \FlecheSymbolic{\edge_0} \symbstate_1\FlecheSymbolic {\edge_1} \cdots \FlecheSymbolic{\edge_{m-1}} \symbstate_m$, such that for all $i = 0, \dots, m-1$, $\edge_i \in \Edges$, and %
$(\symbstate_{i}, \edge, \symbstate_{i+1}) \in {\TransSymb}$.
Given a symbolic state~$\symbstate$, we say that $\symbstate$ is reachable if $\symbstate$ belongs to a symbolic run of~$\A$.

Let $(\locinit, \C) = \Init(\A)$.
Since all basic operators on polyhedra (intersection, future, time past, resets…)\ preserve polyhedra as mentioned in \cref{sss:valuations}, $\C$ is a convex polyhedron.
For the same reason, we can recall the following property:

\begin{propertyC}[{\cite[Property~1]{JLR15}}]\label{property:polyhedron}
	For any symbolic state $(\loc, \C)$ reachable from~$\symbstateinit$ in the parametric zone graph of~$\A$, $\C$ is a convex polyhedron.
\end{propertyC}

\paragraph{Relationship between concrete and symbolic runs}\label{paragraph:symbolic-concrete}
Let us recall the relationship between concrete and symbolic runs, from \eg{}~\cite{JLR15}.

\begin{lemC}[{\cite[Corollary~2]{JLR15}}]\label{cor:next}
	For each parameter valuation $\pval$, reachable symbolic state~$\symbstate$, and concrete state~$\state$,
	we have $\state \in \valuate{\symbstate}{\pval}$ if and only if there is a run of $\valuate{\A}{\pval}$ from the initial state leading to~$\state$.
\end{lemC}

\subsubsection{Integer hulls}\label{sss:IH}
We briefly recall some definitions from~\cite{JLR15}.
Let $\C$ be a convex polyhedron.

The integer hull of a topologically closed polyhedron, denoted by $\IH(\C)$, is defined as the convex hull of the integer vectors of~$\C$, \ie{} $\IH(\C) = \Conv(\IV(\C))$, where $\IV=\{\varval \mid \varval \in\C \land \varval \text{ is a }\setZ\text{-valuation}\}$ denotes the set of vectors with integer coordinates.\label{def:IV}

The interested reader is refered to~\cite{JLR15} for a discussion on how the integer hull operation can be adapted to non-necessarily closed polyhedra.

We treat integer hulls for finite unions of convex polyhedra in a manner similar to~\cite{JLR15}: given a (possibly non-convex) finite union of convex polyhedra $\bigcup_i \C_i$, we write $\IH(\bigcup_i \C_i)$ for the set $\bigcup_i(\IH(\C_i))$.
Given a symbolic state $\symbstate = (\loc, \C)$, we often write $\IH(\symbstate)$ for $(\loc, \IH(\C))$.
\subsubsection{Computation problems}\label{sss:problems}
Given a class of decision problems %
	(such as reachability), we consider the problem of synthesizing the set (or part of it) of parameter valuations $\pval$ such that $\valuate{\A}{\pval}$ satisfies a given problem.
Here, we mainly focus on
reachability (\ie{} ``synthesize the set of parameter valuations for which a (concrete) run visits some goal location'') called~EF,
and unavoidability (\ie{} ``do all maximal (concrete) runs go through some goal locations'') called~AF.\footnote{%
		The EF and AF notations come from TCTL and denote reachability and unavoidability, respectively; they have been used in several works since~\cite{JLR15}.
	}

\defProblem
	{EF-synthesis}
	{A PTA~$\A$ and a subset~$\LocsTarget$ of its locations}
	{Synthesize the set of parameter valuations $\pval$ such that~$\LocsTarget$ is reachable in $\valuate{\A}{\pval}$ from the initial state.}
\defProblem
	{AF-synthesis}
	{A PTA~$\A$ and a subset~$\LocsTarget$ of its locations}
	{Synthesize the set of parameter valuations~$\pval$ such that all concrete runs eventually reach~$\LocsTarget$ in $\valuate{\A}{\pval}$.}

\smallskip

Finally, we will be interested in \cref{ss:IM} in the traces preservation synthesis.

\defProblem
	{Traces preservation-synthesis}
	{A PTA~$\A$ and a parameter valuation~$\pval$}
	{Synthesize the set of parameter valuations~$\pval'$ such that $\valuate{\A}{\pval'}$ has the same trace set as~$\valuate{\A}{\pval}$.}

\smallskip

Since the stated problems are not computable in general, our goal here will be to synthesize sets of parameters containing all integer valuations solution to the problem, and possibly some rational valuations too.\label{newtext:goal}
\section{Parametric extrapolation}\label{section:extrapolation}

In this section, we present an extrapolation based on the classical $\LargestC$-extrapolation used for the symbolic state abstraction for timed automata (see \eg{}~\cite{DT98}), but this time in a parametric setting.
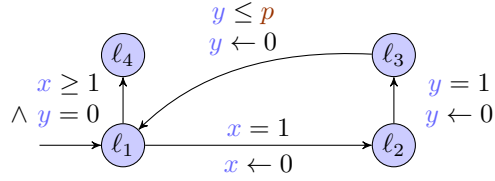
\begin{figure}[tb]
	\centering
    \small
	\begin{tikzpicture}[PTA, scale=1.2, xscale=1.5]
 
		\node[location, initial] at (0,0) (l1) {$\loc_1$};
 
		\node[location] at (2,0) (l2) {$\loc_2$};
 
		\node[location] at (2,1) (l3) {$\loc_3$};
 
		\node[location] at (0, 1) (l4) {$\loc_4$};

		\path (l1) edge[below ] node[above]{$ \clockx = 1$} node[below]{$\clockx\assign0$} (l2);

		\path (l1) edge node{\begin{tabular}{c @{\ } c}
		& $ \clockx \geq 1$
		\\
		$\land$ & $ \clocky = 0$
		\end{tabular}} (l4);

		\path (l2) edge[right] node{\begin{tabular}{c @{\ } c}
		& $ \clocky = 1$\\
		 & $\clocky\assign0$\\
		\end{tabular}} (l3);

		\path (l3) edge[above,bend right] node{\begin{tabular}{c @{\ } c}
		& $ \clocky \leq \styleparam{p}$\\
		 & $\clocky\assign0$\\
		\end{tabular}} (l1);
 
	\end{tikzpicture}
	
	\caption{Motivating PTA}
	\label{figure:example:reviewerACSD}

\end{figure}

In~\cite{JLR15}, termination was guaranteed because it was possible to bound all clock valuations using a bound depending only on the integer constants of the PTA and on the bounded parameter domain, which in turns allow to bound the value of all clocks, and finally the number of possible symbolic states.
Let us first show that this does not hold anymore over rational-valued parameters, which motivates the use of an extrapolation.

\begin{example}\label{example:motivation-extrapolation}
Consider the PTA in \cref{figure:example:reviewerACSD}.
The smaller the value of~$\param$, the more loops we need to take, and therefore the longer the time to reach~$\loc_4$:
after a number~$n$ of times through the loop, when arriving in $\loc_1$ from $\loc_3$ we get constraints in~$\loc_1$ of the form
$ 0 \leq x \leq n \times \param$, with $n$ growing without bound. Indeed assume that $x\leq n\param$ and $y=0$ in $\loc_1$. Then $x = 0$ and $y\geq 1 -n\param$ in $\loc_2$, then $x\leq n\param$ in $\loc_3$ and finally $x\leq (n+1)\param$ and $y=0$ back in $\loc_1$. 
Since the duration of the loop (between two consecutive arrivals in~$\loc_2$) is exactly~1, then the time necessary to reach location~$\loc_4$ is arbitrarily large, depending on the value of~$\param$---even if the parameter~$\param$ is bounded (\eg{} in $[0,1]$).
This was not the case in~\cite{JLR15} due to the fact that parameters were bounded integers.
Hence, on this PTA, we cannot just apply the integer hull (as in~\cite{JLR15}) to ensure termination of our algorithms.
\end{example}
\begin{example}\label{example:k-nonconvex}
 Now, we will show that the union for all valuations of the parameters  of the classical $\LargestC$-extrapolation used for the classical symbolic state abstraction for timed automata~\cite{DT98} leads to a non-convex polyhedron.
Let us consider the PTA in \cref{figure:extrapo:pta} with a parameter~$p$ such that $0 \leq \param \leq 1$.
	Since when~$x$ is reset to~$0$ by taking the loop, $y$ has increased by at most~$\param$ due to the invariant, then by taking $n$~times the loop we obtain:
	$$
	0\leq x\leq \param \ \  \land \ \ 
	0\leq y-x\leq (n+1)\times \param \ \  \land \ \ 
	0 \leq \param \leq 1
	$$

Recall that the classical $\LargestC$-extrapolation for timed automata (see~\cite{DT98,BBLP06}) consists in exhibiting the largest constant of the model (``$\LargestC$''): then, each guard is either always true or always false for any clock valuation beyond this threshold.
Extrapolation consists in adding (spurious) concrete states beyond~$\LargestC$ such that extrapolated polyhedra overlapping~$\LargestC$ have a finite number of possible shapes---which ensures the finiteness of the (extrapolated) symbolic state graph, and therefore termination of the related symbolic algorithms.

Here,
the greatest  constant of the model is $\LargestC = 1$.
	After one loop (and some delay), $y$ can be greater than~$1$.
	Then, for each value of~$p$, we can apply to the polyhedron given above the classical $\LargestC$-extrapolation for TAs.
	The union for all valuations of $p$ of these extrapolations, projected to the plan $(y,p)$, is depicted by the colored part (light orange and dark blue) of \cref{figure:extrapo:extrapo}.
	The right-most polyhedron (dark blue) is unbounded in~$y$ (depicted by the gradient).
	Note that the obtained extrapolation is non-convex.
	This is similar to the non-convex approximations for TAs~\cite{HKSW11}.
\end{example}
	\begin{figure}[tb]
	\centering
	\small

	\begin{subfigure}[b]{0.4\textwidth}
		\centering
		\begin{tikzpicture}[PTA, thin]

			\node[location, initial] at(-2,1) (l0) {$\loc_0$};
			\node[location] at(0,1) (l1) {$\loc_1$};
			\node [invariant, below] at (-2,0.5) {$\clockx \leq \paramp$};

			\path
				(l0) edge[loop] node [above] {$\clockx \assign 0$} (l0)
				(l0) edge[] node {$\clocky > 1$}  (l1)
			;
			\end{tikzpicture}
			
		\caption{PTA (assume $\ParamDomain(\paramp) = [0, 1]$)}
		\label{figure:extrapo:pta}
		\end{subfigure}
		\hfill
	\begin{subfigure}[b]{0.4\textwidth}
		\centering
		\begin{tikzpicture}[PTA, thin]
			\tikzstyle{axe} = [line width=1pt, ->, draw=black!80]
			\tikzstyle{fondgris} = [fill=black!5, draw=none]
			\tikzstyle{zone} = [fill=blue!40!white, draw=blue!60!white,line width=1pt]
			\tikzstyle{openzone1} = [fill=orange!20!white, draw=none]
			\tikzstyle{openzone2} = [fill=blue!60!white, draw=none]

			\draw[openzone1] (2, 1) -- (3, 1.5) -- (3, 2)-- (2, 2) -- cycle;
			\draw[openzone2] (3, 1.5) -- (5, 1.5) -- (5, 2) -- (3, 2) -- cycle;
			\draw[left color=blue!60!white,draw=none] (5, 2) rectangle (5.5, 1.5);

			\path[axe] %
				(2, 1) -- ++ (3.5, 0);
			\path[axe] %
				(2, 1) -- ++ (0, 1.5);

			\node at (5.7,1) {$\clocky$};
			\node at (2,2.7) {$\paramp$};

			\foreach \x in {0, 1, ..., 3} %
				\draw [-](\x+2, 1) -- (\x+2, 0.8);
			\foreach \x in {0, 1, ..., 3} %
				\node at (\x+2,0.6) {\x};
			\foreach \x in {0, 1} %
				\draw [-](2, \x+1) -- (1.8, \x+1); %
			\foreach \x in {0, 1} %
				\node at (1.5,\x+1) {\x};

			\end{tikzpicture}
			
		\caption{Extrapolation}
		\label{figure:extrapo:extrapo}
		\end{subfigure}
		
	\caption{Example illustrating the non-convex parametric extrapolation}
	\label{figure:extrapo}
	
	\end{figure}
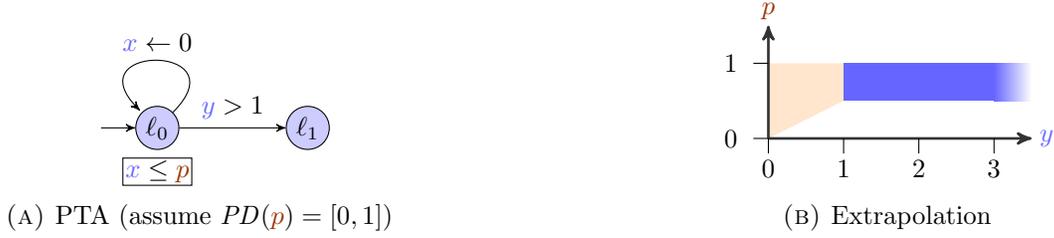
\subsection{A parametric extrapolation}\label{ss:parametric-extrapolation}

For any polyhedron $\C$ over a given set of variables~$\Variables$, for any variable $\variable \in \Variables$, we denote by $\cylinder{\variable}(\C)$ the \emph{cylindrification} of $\C$ along variable~$\variable$, \ie{} $\cylinder{\variable}(\C)=\{\varval \mid \exists \varval' \in \C, \forall \variable' \neq \variable, \varval'(\variable') = \varval(\variable')\text{ and } \varval(\variable) \geq 0\}$.
This is a usual operation that consists in \emph{unconstraining} variable~$\variable$.
The result $\cylinder{\variable}(\C)$ for a polyhedron~$\C$ is a polyhedron: one just need to project out variable~$\variable$ (\eg{} with the Fourier-Motzkin algorithm), and then add it back unconstrained.

\subsubsection{Defining extrapolation}\label{sss:extrapolation-definitions}
Let us first introduce our concept of $(\LargestC, \clock)$-extrapolation for a single clock.

\begin{definition}[$(\LargestC,\clock)$-extrapolation]
    \label{def:extx}
	Let $\C$ be a polyhedron. Let $\LargestC$ be a non-negative integer constant and $\clock$ be a clock.
	The $(\LargestC,\clock)$-extrapolation of $\C$, denoted by $\Ext{\LargestC}{\clock}(\C)$, is defined as:
	$$\Ext{\LargestC}{\clock}(\C)= \big(\C\cap \valuations{\clock \leq \LargestC}\big) \cup \Big( \cylinder{\clock}\big(\C\cap \valuations{\clock>\LargestC}\big)\cap \valuations{\clock>\LargestC} \Big)\text{.}$$
\end{definition}

Given $\symbstate = (\loc, \C)$, we write $\Ext{\LargestC}{\clock}(\symbstate)$ for $\big( \loc, \Ext{\LargestC}{\clock}(\C) \big)$.

\begin{example}
 To illustrate the $(\LargestC, \clock)$-extrapolation, we go back to the example of \cref{figure:extrapo:pta}.
 After one loop, the set~$\C$ of valuations associated with $\locinit$ is
	$\clockx \leq \paramp \land \clocky - \clockx \leq \paramp \land \clocky \geq \clockx \geq 0$.

According to \cref{def:extx}, $\Ext{1}{y}(\C)$ is made of two parts:
 \begin{itemize}
 	\item $\big(\C\cap \valuations{y \leq 1}\big)$ evaluates to $0 \leq \clockx \leq \clocky \leq 1 \land \clocky - \paramp \leq \clockx \leq \paramp$, depicted in light orange in \cref{figure:extrapo:extrapo} after projection onto $\clocky$ and~$\paramp$ (\ie{} $0 \leq \clocky \leq 1 \land \clocky \leq 2 \times \paramp$).

 	\item $\cylinder{y}\big(\C\cap \valuations{y>1}\big)\cap \valuations{y>1}$ evaluates to $0 \leq \clockx \leq \paramp \leq 1 \land \clocky > 1 \land \clockx + \paramp > 1$, depicted in dark blue in \cref{figure:extrapo:extrapo} after projection onto $\clocky$ and~$\paramp$ (\ie{} $\clocky > 1 \land \frac{1}{2} < \paramp \leq 1$).

 \end{itemize}

 Note that for this example the $(1,y)$-extrapolation gives the same result as the union for all valuations of the parameter $\param$ of the classical extrapolation for timed automata.
\end{example}

\cref{lemma:extxy} follows from \cref{def:extx}.

\begin{lemma}
    For each polyhedron $\C$, integer $\LargestC\geq 0$ and clock variables $\clock$ and $\clock'$, we have $\Ext{\LargestC}{\clock}\big(\Ext{\LargestC}{\clock'}(\C)\big)=\Ext{\LargestC}{\clock'}\big(\Ext{\LargestC}{\clock}(\C)\big)$.
    \label{lemma:extxy}
\end{lemma}
\begin{proof}%
	The result comes from the following facts:
	\begin{enumerate}
		\item $\cylinder{\clock}\big(\cylinder{\clock'}(\C)\big) = \cylinder{\clock'}\big(\cylinder{\clock}(\C)\big)$;
		\item for $\clock \neq \clock', \cylinder{\clock}(\C) \cap \valuations{\clock' \bowtie \LargestC} = \cylinder{\clock}\big(\C \cap \valuations{\clock' \bowtie \LargestC}\big)$ for ${\bowtie} \in \{ <,\leq,\geq,> \}$.
		\qedhere
	\end{enumerate}
\end{proof}

We can now consistently define the $(\LargestC,\Clock)$-extrapolation operator over a \emph{set} of clocks:

\begin{definition}[$(\LargestC,\Clock)$-extrapolation]\label{definition:MX-extrapolation}
    Let $\LargestC$ be a non-negative integer constant and $\Clock$ be a set of clocks.
    The $(\LargestC,\Clock)$-extrapolation operator $\Ext{\LargestC}{\Clock}$ is defined as the composition (in any order) of all $\Ext{\LargestC}{\clock}$, for all $\clock\in\Clock$.
\end{definition}
\begin{remark}\label{remark:size-ext}
	Observe that, given a convex polyhedron~$\C$, the size of $\Ext{\LargestC}{\Clock}(\C)$ is in the worst case a disjunction of $2^{|\Clock|}$ polyhedra.
\end{remark}

When clear from the context we omit $\Clock$ and only write $\LargestC$-extrapolation or $\Ext{\LargestC}{}$.

In the rest of this section, we first prove our results on $\Ext{\LargestC}{\clock}$.
It is then straightforward to adapt them to $\Ext{\LargestC}{\Clock}$ using \cref{lemma:extxy}.

\subsubsection{Properties of extrapolation}\label{sss:extrapolation-results}
Crucially, extrapolation preserves the projection onto~$\Param$:

\begin{lemma}\label{lemma:extproj}
	Let $\C$ be a constraint over $\Clock \cup \Param$.
	Then $\projectP{\C}=\projectP{\Ext{\LargestC}{\clock}(\C)}$.
\end{lemma}

\begin{proof}
	First notice that $\C\subseteq \Ext{\LargestC}{\clock}(\C)$, and therefore $\projectP{\C}\subseteq\projectP{\Ext{\LargestC}{\clock}(\C)}$. 
	Second, $\Ext{\LargestC}{\clock}$ only adds points with new clock valuations, without modifying parameter valuations.
	Indeed, the left-hand term only constrains~$\C$ further and therefore does not add parameter valuations.
	The right-hand part first constrains~$\C$, then unconstrains~$\clock$ (which does not add parameter valuations), and constrains~$\C$ again: so no new parameter valuations can be added.\label{newtext:lemma:extproj-noadd}
	By definition of $\Ext{\LargestC}{\clock}$, for any valuation in $\Ext{\LargestC}{\clock}(\C)$, there exists a valuation in $\C$ with the same projection on parameters. So, $\projectP{\C}\supseteq\projectP{\Ext{\LargestC}{\clock}(\C)}$.
\end{proof}

\begin{lemma}\label{lemma:ALR15-lemma4}
	For each parameter valuation $\pval$, non-negative integer constant $\LargestC$, clock~$\clock$ and valuations set~$\C$,
		$\valuate{\Ext{\LargestC}{\clock}(\C)}{\pval} = \Ext{\LargestC}{\clock}(\valuate{\C}{\pval})$.
   \label{lemma:extv}
\end{lemma}

\begin{proof}
    It is easy to prove, just by writing their definitions, that $\valuate{\C\cap\C'}{\pval}=\valuate{\C}{\pval}\cap\valuate{\C'}{\pval}$ and $\cylinder{\clock}(\valuate{\C}{\pval})=\valuate{\cylinder{\clock}(\C)}{\pval}$, and the result follows.
\end{proof}

We will also need the following result:
\begin{lemma}
For all edges $\edge$ and sets of valuations $\C_1$ and $\C_2$, we have: if $\Ext{\LargestC}{\clock}(\C_1) = \Ext{\LargestC}{\clock}(\C_2)$ then 
$\Ext{\LargestC}{\clock}(\Succ(\C_1,\edge)) = \Ext{\LargestC}{\clock}(\Succ(\C_2,\edge))$.
	\label{lemma:extsucc}
\end{lemma}

\begin{proof}
	Assume $\Ext{\LargestC}{\clock}(\C_1) = \Ext{\LargestC}{\clock}(\C_2)$
	and consider $\wv{\clockval_1'}{\pval}\in\Succ(\C_1,\edge)$.
	Then there exists $\wv{\clockval_1}{\pval}\in \C_1$ such that
	$\wv{\clockval_1'}{\pval}\in \Succ(\{\wv{\clockval_1}{\pval}\}, \edge)$.
	
	If $\valuate{\clock}{\clockval_1} \leq \LargestC$, then $\wv{\clockval_1}{\pval}$ is also in $\C_2$ and 
	$\wv{\clockval'_1}{\pval}\in \Succ(\C_2, \edge)$.

	If $\valuate{\clock}{\clockval_1} > \LargestC$, then by definition of $\Ext{\LargestC}{\clock}$ and in particular of cylindrification, there exists some $\clockval_2\in\C_2$ such that $\forall \clock'\neq \clock, \valuate{\clock'}{\clockval_1} = \valuate{\clock'}{\clockval_2}$ and $\valuate{\clock}{\clockval_2} > M$.  

	Now, either $\clock$ is reset to zero along $\edge$, and then it is clear that $\Succ(\{\wv{\clockval_1}{\pval}\}, \edge) = \Succ(\{\wv{\clockval_2}{\pval}\}, \edge)$ and so $\wv{\clockval'_1}{\pval}\in \Succ(\C_2, \edge)$.	

	Or $\clock$ is not reset to zero along $\edge$.
	If we have lower bound constraints on $\clock$ in guards or invariants, those constraints were and are vacuously satisfied before and after taking $\edge$, both for $\clockval_1$ and $\clockval_2$.  And if we have any upper bound constraint on $\clock$ then the successor is empty, again for both valuations. If the edge can indeed be taken (the successor is not empty), then clock $\clock$ thus plays no role in determining the value of other clocks and itself will stay above $M$.

	It follows we still have $\clockval'_1(\clock) > M$ and there also exists $\clockval'_2\in\Succ(\{\wv{\clockval_2}{\pval}\}, \edge)$, such that  $\clockval'_2(\clock) > M$ and for all $\clock'\neq\clock,  \clockval'_2(\clock') = \clockval'_1(\clock')$. This means that $\wv{\clockval'_1}{\pval}\in \Ext{\LargestC}{\clock}(\Succ(\C_2, \edge))$.

	With this, we have proved the left to right inclusion. The right to left inclusion is completely symmetrical.
\end{proof}

It is clear that \cref{lemma:extproj,lemma:extv,lemma:extsucc} directly extend from $\Ext{\LargestC}{\clock}$ to $\Ext{\LargestC}{\Clock}$.

\subsection{Simulations}

For the preservation of behaviours, following~\cite{BBLP06}, we use a notion of simulation:
\begin{definition}[Simulation~\cite{BBLP06}]\label{definition:simulation-BBLP06}
    Let \mbox{$\A = (\Actions, \Loc, \locinit, \Clock, \invariant, \Edges)$} be a TA and $\preceq$ a relation on $L\times\setRplus^{\cardinality{\Clock}}$. Relation $\preceq$ is a (location-based) simulation if:
    \begin{itemize}
        \item if $(\loc_1,\clockval_1)\preceq (\loc_2,\clockval_2)$ then $\loc_1=\loc_2$,
        \item if $(\loc_1,\clockval_1)\preceq (\loc_2,\clockval_2)$ and $(\loc_1,\clockval_1)\fleche{\action}(\loc'_1,\clockval'_1)$, then there exists $(\loc'_2,\clockval'_2)$ such that $(\loc_2,\clockval_2)\fleche{\action}(\loc'_2,\clockval'_2)$ and $(\loc'_1,\clockval'_1)\preceq (\loc'_2,\clockval'_2)$,
        \item if $(\loc_1,\clockval_1)\preceq (\loc_2,\clockval_2)$ and $(\loc_1,\clockval_1)\fleche{d_1}(\loc_1,\clockval_1 + d_1)$, then there exists $d_2$ such that $(\loc_2,\clockval_2) \fleche{d_2} (\loc_2,\clockval_2 + d_2)$ and $(\loc_1,\clockval_1 + d_1)\preceq (\loc_2,\clockval_2 + d_2)$.
    \end{itemize}

    If $\preceq^{-1}$ is also a simulation relation then $\preceq$ is called a bisimulation.

    State $\state_1$ simulates $\state_2$ if there exists a simulation $\preceq$ such that $\state_2\preceq \state_1$. If $\preceq$ is a bisimulation, then the two states are said bisimilar.
\end{definition}

\begin{lemC}[{\cite[Lemma~1]{BBLP06}}]\label{lemma:bisim}
	Let $\LargestC$ be a non-negative integer constant greater than or equal to the maximum constant occurring in the simple clock constraints appearing in the guards and invariants of the TA.
	Let $\equiv_\LargestC$ be the relation defined as $\clockval\equiv_\LargestC \clockval'$ iff $\forall \clock\in \Clock$: either $\valuate{\clock}{\clockval}=\valuate{\clock}{\clockval'}$ or ($\valuate{\clock}{\clockval} > \LargestC$ and $\valuate{\clock}{\clockval'}>\LargestC$).
The relation $\mathcal{R} = \big \{ \big ((\loc, \clockval), (\loc, \clockval') \big ) \mid \clockval \equiv_\LargestC \clockval' \big \}$ is a bisimulation relation.

\end{lemC}

\subsubsection{Computing $\LargestC$ in bounded PTAs}

We use the bounds on parameters to compute the maximum constant~$\LargestC$ appearing in all the guards and invariants of the PTA.
When in a constraint a clock is compared to a parametric expression, we compute the maximum value of that expression over all the bounded parameter valuations. This value is finite and can be done by solving a linear program because the expression is linear and the domain of the parameters defines a polytope.
$\LargestC$ should be greater than all these valuations.

\begin{example}\label{example:computing-M}
	Consider a guard $\clock\leq 2\paramitext{1} -\paramitext{2}+1$ and $\paramitext{1}\in [2,5]$, and $\paramitext{2}\in [3,4]$; then the maximum constant corresponding to this constraint is $2\times 5-3+1=8$.
\end{example}

Also note that bounding the parameter domain of PTAs is not a strong restriction in practice---especially since the bounds can be arbitrarily large.\footnote{%
	Of course, the theoretical worst-case complexity of model checking depends exponentially on the size of the constants;
	however, appropriate abstractions such as the zone graph~\cite{BY03} are well-known to potentially alleviate this problem in practice. %
}

\cref{lemma:extbisim,lemma:extintv2} are instrumental in proving the preservation of all correct integer parameter valuations in the algorithms of \cref{section:algos}.

\begin{lemma}\label{lemma:extbisim}
    Let~$\A$ be a bounded PTA, $\symbstate$ be a symbolic state of~$\A$, and $\LargestC$ a non-negative integer constant greater than the maximal constant occurring in the clock constraints of the PTA for all possible valuations of the parameters.
    Let $\clock$ be a clock, $\pval$ be a parameter valuation, and $(\loc, \clockval) \in \valuate{\Ext{\LargestC}{\clock}(\symbstate)}{\pval}$ be a concrete state.
    There exists a state $(\loc, \clockval') \in \valuate{\symbstate}{\pval}$ such that $(\loc, \clockval)$ and $(\loc, \clockval')$ are bisimilar.
\end{lemma}

\begin{proof}
    If $(\loc, \wv{\clockval}{\pval}) \in \symbstate$, then the results holds trivially.
	Otherwise, it means that there exists some clock $\clock$ such that 
    $(\loc, \wv{\clockval}{\pval}) \in \cylinder{\clock}(\symbstate\cap \valuations{\clock>\LargestC})\cap\valuations{\clock>\LargestC}$. This implies that
		$\valuate{\symbstate\cap \valuations{\clock>\LargestC}}{\pval}\neq \emptyset$ and $\valuate{\clock}{\clockval}>\LargestC$.
        Therefore, and using the definition of $\cylinder{\clock}$, there exists $(\loc, \wv{\clockval'}{\pval}) \in \symbstate\cap \valuations{\clock>\LargestC}$ such that for each $\clock'\neq \clock,\valuate{\clock'}{\clockval'}=\valuate{\clock'}{\clockval}$. We also have $\valuate{\clock}{\clockval'}>\LargestC$, which means that $\clockval'\equiv_\LargestC \clockval$ and by \cref{lemma:bisim}, we obtain the expected result.
\end{proof}

Again, \cref{lemma:extbisim} directly extends from $\Ext{\LargestC}{\clock}$ to $\Ext{\LargestC}{\Clock}$.
The following lemma will be subsequently used in the proof of soundness of our synthesis algorithm (\cref{theorem:RIEF:soundness}).

\begin{lemma}\label{lemma:extbisimeq}
    For all symbolic states $\symbstate$ and $\symbstate'$, non-negative integer constant $\LargestC$ greater than the maximal constant occurring in the PTA (including the bounds of parameters), and parameter valuations $\pval$, such that $\valuate{\Ext{\LargestC}{\Clock}(\symbstate)}{\pval}=\valuate{\Ext{\LargestC}{\Clock}(\symbstate')}{\pval}$, for all states $(\loc, \clockval) \in \valuate{\symbstate}{\pval}$, there exists a state $(\loc, \clockval') \in \valuate{\symbstate'}{\pval}$ such that $(\loc, \clockval)$ and $(\loc, \clockval')$ are bisimilar.
\end{lemma}

\begin{proof}
    This is a direct consequence of \cref{lemma:extbisim,lemma:extv}. %
\end{proof}

\subsection{Extrapolation and integer hulls}

Here, for the sake of simplicity, and similarly to~\cite{JLR15}, we consider that all polyhedra are topologically closed and, to avoid confusion, we equivalently (provided that $\LargestC$ is (strictly) greater than the maximal constant in the PTA) define $\Ext{\LargestC}{\clock}(\symbstate)$ as $(\symbstate\cap\valuations{\clock\leq \LargestC})\cup\cylinder{\clock}(\symbstate\cap\valuations{\clock\geq \LargestC})\cap\valuations{\clock\geq \LargestC}$.

\begin{lemma}\label{lemma:extintv}
    For any integer parameter valuation~$\pval$, any non-negative integer constant~$\LargestC$, and any symbolic state $\symbstate=(\loc, \C)$, $\valuate{\Ext{\LargestC}{\clock}(\C)}{\pval}$ is its own integer hull.
\end{lemma}
	
\begin{proof}
	For any reachable symbolic state $(\loc,\C)$ of a PTA $\A$, and for any integer parameter valuation~$\pval$, we have by~\cite[Property~3]{JLR15} that
	$\valuate{\pval}{\A}$ is a \emph{zone}~\cite{BY03}: a convex polyhedron defined only by inequalities of the form $\variable \bowtie d$ or $\variable -\variable'\bowtie d$, with $\variable,\variable'$ (clock) variables, $d\in \setZ$, and ${\bowtie} \in \{<,\leq,\geq,>\}$.
	Still by~\cite[Property~3]{JLR15}, zones also have integer vertices and are their own integer hulls.
	
    From \cref{lemma:extv}, $\valuate{\Ext{\LargestC}{\clock}(\C)}{\pval} = \Ext{\LargestC}{\clock}(\valuate{\C}{\pval})$. 
	Now, $\valuations{\clock \geq \LargestC}$ and $\valuations{\clock\leq \LargestC}$ are zones too. Furthermore, cylindrification on $\clock$ preserves zones with integer coefficients too: it consists in
    removing any upper bound constraint on $\clock$ ($\clock\leq U$%
    , with $U$ an integer) from the zone.
    Thus $\Ext{\LargestC}{\clock}(\valuate{\C}{\pval})$ is a union of two zones, and each of them is then its own integer hull.
    The result then follows from the fact that the integer hull of a union of polyhedra is defined as the union of the integer hulls of those polyhedra.
\end{proof}

\begin{lemma}\label{lemma:extintv2}
	 For any integer parameter valuation~$\pval$, any non-negative integer constant~$\LargestC$, and any symbolic state $\symbstate=(\loc, \C)$, $\valuate{\IH(\Ext{\LargestC}{\Clock}(\C))}{\pval}=\valuate{\Ext{\LargestC}{\Clock}(\C)}{\pval}$.
\end{lemma}

\begin{proof}
    Since we take the integer hull on each convex parts, it certainly holds that
    $\IH(\Ext{\LargestC}{\Clock}(\symbstate))\subseteq\Ext{\LargestC}{\Clock}(\symbstate)$ and so
    $\valuate{\IH(\Ext{\LargestC}{\Clock}(\C))}{\pval}\subseteq\valuate{\Ext{\LargestC}{\Clock}(\C)}{\pval}$.

    In the other direction, using \cref{lemma:extintv}, $\valuate{\Ext{\LargestC}{\Clock}(\C)}{\pval}=\IH(\valuate{\Ext{\LargestC}{\Clock}(\C)}{\pval})$. Now, if $\clockval\in\IH(\valuate{\Ext{\LargestC}{\Clock}(\C)}{\pval})$, then $\wv{\clockval}{\pval}\in\IH(\Ext{\LargestC}{\Clock}(\C))$, \ie{} $\clockval\in \valuate{\IH(\Ext{\LargestC}{\Clock}(\C)}{\pval}$. %
\end{proof}

The following \cref{lemma:finite} is the key to proving the termination of the algorithms of \cref{section:algos}.

\begin{proposition}\label{lemma:finite}
	In a bounded PTA, the set of constraints $\IH(\Ext{\LargestC}{\Clock}(\C))$ over the set of symbolic reachable states $(\loc, \C)$ is finite.
\end{proposition}
\begin{proof}
    In each disjunct of $\Ext{\LargestC}{\Clock}(\C)$,
		either clocks are upper bounded by $\LargestC$, or (due to the cylindrification) are lower bounded by~$\LargestC$.
    Therefore, vertices of these disjuncts all have coordinates less or equal to~$\LargestC$.
    When taking the integer hull of all these disjuncts separately we obtain a finite union of polyhedra with integer vertices with coordinates less or equal to~$\LargestC$ (and a finite set of extremal rays, taken in the finite set of the directions of clock variables), of which there can be only finitely many. %
\end{proof}

\section{Integer-complete dense parametric algorithms}\label{section:algos}
In this section, we describe three parameter synthesis algorithms:
\begin{enumerate}
	\item \RIEF{}, computing parameter valuations for which a location is reachable (\cref{ss:RIEF});
	\item \RIAF{}, computing parameter valuations for which a location is unavoidable (\cref{ss:RIAF}); and
	\item \RITPS{}, computing parameter valuations for which the trace set is identical to that of a reference valuation (\cref{ss:IM}).
\end{enumerate}

These algorithms always terminate for \emph{bounded} PTAs, and return not only all the integer valuations solution of the problem (à la~\cite{JLR15}) but also some rational valuations.
In fact, contrarily to~\cite{JLR15} in which returned constraints could only be interpreted over \emph{integer} valuations, all rational valuations in the constraints returned by our algorithms are correct solutions.\label{newtext:integer-rational-valuations}

\paragraph{Principle of the three algorithms}\label{paragraph:principle}
All three algorithms (\RIEF{}, \RIAF{} and \RITPS{}) explore the standard parametric symbolic state graph (without integer hull nor extrapolation); integer hull and extrapolation are only used in the $\Passed$ set of visited states, to ensure termination.
This is in contrast with most algorithms for timed automata, which work on extrapolated zone graphs.

In order to prove the soundness and completeness of the algorithms, we inductively define, as in~\cite{JLR15}, the \emph{symbolic reachability tree} of~$\A$ as the possibly infinite directed labelled tree $\symtree$ such that:
\begin{itemize}
    \item the root of~$\symtree$ is labelled by $\Init(\A)$;
    \item for every node $n$ of~$\symtree$, if $n$ is labelled by some symbolic state $\symbstate$, then for all edges~$\edge$ of~$\A$, there exists a unique child $n'$ of~$n$ labelled by $\Succ(\symbstate, \edge)$ iff $\Succ(\symbstate, \edge) \neq \emptyset$.
\end{itemize}

\begin{remark}
	\cref{def:boundedPTA} defines bounded PTAs as PTAs in which the parameters range in a hyperrectangle.
	However, all results in this section can be extended to a more permissive definition of bounded PTAs.
	In fact, in this section, we could consider that a PTA is bounded if its parameter domain is bounded.
    This includes hyperrectangles, but also more free ``shapes'', including non-convex or non-connected parameter domains, provided they can be expressed as a finite union of convex polyhedra.
\end{remark}
\subsection{Parametric reachability: \RIEF}\label{ss:RIEF}
\begin{algorithm}[tb!]
	\SetKwInOut{Input}{input}
	\SetKwInOut{Output}{output}

	\Input{A bounded PTA $\A$, a symbolic state $\symbstate=(\loc, \C)$, a set of target locations~$\LocsTarget$, a set~$\Passed$ of passed states on the current path}
	\Output{Set~$\K$ of parameter valuations guaranteeing reachability}

	\BlankLine

	\lIf{$\loc \in \LocsTarget$}{%
		$\K \assign \projectP{\C}$\nllabel{algo:RIEF:line2}
	}\Else{%
		$\K \assign \emptyset $\;
        \If{\hl{$\IH$}$\big($\hl{$\Ext{\LargestC}{\Clock}$}$(\symbstate) \big) \notin \Passed$}{%
			\ForEach{outgoing $\edge$ from $\loc$ in $\A$ s.t.\ $\Succ(\symbstate, \edge) \neq \emptyset$}{%
				$\symbstate' \assign \Succ(\symbstate, \edge)$\;

				$\K \assign \K \cup \RIEF\Big(\A, \symbstate', \LocsTarget, \Passed \cup \big\{ $\hl{$\IH$}$($\hl{$\Ext{\LargestC}{\Clock}$}$(\symbstate)) \big\}\Big)$\;
			}
		}
	}
	
	\Return{$\K$}
	\caption{$\RIEF(\A, \symbstate, \LocsTarget, \Passed)$}
	\label{algo:RIEF} %
\end{algorithm}

The goal of \RIEF{} given in \cref{algo:RIEF} (``R'' stands for robust, and ``I'' for integer hull) is to synthesize parameter valuations solution to the EF-synthesis problem, \ie{} the valuations for which there exists a run eventually reaching a location in~$\LocsTarget$.
It is inspired by the algorithms \EF{} and \IEF{} introduced in~\cite{JLR15} that both address the same problem; however \EF{} does not terminate in general, and \IEF{} can only output integer valuations.
In fact, if we replace all occurrences of $\IH(\C)$ in Algorithm $\RIEF$ by~$\C$, we obtain Algorithm $\EF$ from~\cite{JLR15}.
The differences with \IEF{} (which was defined in~\cite{JLR15}) are highlighted.

\RIEF{} is a recursive algorithm (initially called by $\RIEF(\A,\Init(\A),\LocsTarget,\emptyset)$) that proceeds as a post-order traversal of the symbolic reachability tree, and collects all parametric constraints associated with the target locations~$\LocsTarget$.

In contrast to \EF{}, \RIEF{} stores in the set~$\Passed$ of visited states the \emph{integer hulls} of (the extrapolation of) the symbolic states, which ensures termination due to the finite number of possible integer hulls of $\LargestC$-extrapolations;
however, in contrast to \IEF{}, \RIEF{} returns the actual states (instead of their integer hull), which yields a larger result than \IEF{}.

As a direct consequence of \cref{lemma:finite}, it is clear that \RIEF{} explores only a finite number a symbolic states.
Therefore, we have the following theorem:
\newcommand{\enonceTheoremRIEFtermination}{
    For any bounded PTA $\A$, the computation of $\RIEF(\A,\Init(\A),\LocsTarget,\emptyset)$ terminates.
}
\begin{theorem}\label{theorem:RIEF:termination}
	\enonceTheoremRIEFtermination{}
\end{theorem}

Algorithm $\RIEF$ is a post-order depth-first traversal of some prefix of the symbolic reachability tree.

\begin{theorem}\label{theorem:RIEF:soundness}
	Let $\K = \RIEF(\A,\Init(\A),\LocsTarget,\emptyset)$.
	We have:
	\begin{enumerate}
		\item Soundness: If $\pval \in \K$ then $\LocsTarget$ is reachable in $\valuate{\A}{\pval}$;
		\item Integer completeness: If $\pval$ is an integer parameter valuation and $\LocsTarget$ is reachable in $\valuate{\A}{\pval}$, then $\pval \in \K$.
	\end{enumerate}
\end{theorem}
\begin{proof}
Since \RIEF{} has terminated, it has explored a finite prefix $\treeprefix$ of~$\symtree$.
\begin{enumerate}
	\item Soundness: 
	Let $\symbstate$ be symbolic state in $\treeprefix$, and $\Passed$ be the set of ancestors of $\symbstate$ in $\treeprefix$. Further let $\Passed'$ be the set $\{(\loc,\IH(\Ext{\LargestC}{\Clock}(\C))\,\mid\, (\loc,\C)\in \Passed\}$.

	Then, by induction, we have $\RIEF(\A, \symbstate, \LocsTarget, \Passed')\subseteq \EF(\A, \symbstate, \LocsTarget, \Passed)$:
	\begin{itemize}
		\item if $\symbstate=(\loc,\C)$ is a leaf of~$\treeprefix$, then either $\loc\in\LocsTarget$ and both algorithms return the same value $\projectP{\C}$, or the result of \RIEF{} is empty and thus included in whatever is the result of \EF{}.
		\item if it is not a leaf, then by the induction hypothesis we have inclusion for all the recursive calls and also for the union of all of them.
    \end{itemize}

	\item Integer completeness: 
		The proof of this part follows the same general structure as that of \EF{} in~\cite{JLR15} but with additional complications due to the use of the extrapolation.
		For the sake of clarity, we rewrite it completely here, taking care of the modified convergence scheme.

		Let $\pval$ be an \emph{integer} parameter valuation such that there exists a run~$\varrun$ in $\valuate{\A}{\pval}$ that
		reaches $\LocsTarget$.  Then $\varrun$ is finite and its last state has
		a location belonging to~$\LocsTarget$.  Let $\edge_1,\ldots, \edge_p$
		be the edges taken in $\varrun$ and consider the branch $\Init(\A)=\symbstate_0\xrightarrow{\edge_1} \symbstate_1 \xrightarrow{\edge_2}\cdots\xrightarrow{\edge_p} \symbstate_p$ of the tree
		$\symtree$ obtained by following this edge sequence on the labels of the
		arcs in the tree. For $i\leq p$, let $\Passed_i = \{\IH(\Ext{\LargestC}{\Clock}(\symbstate_j))\,\mid\, 0\leq j < i\}$.  
		
		Assume first there exists no $k\leq p$ such that $\IH(\Ext{\LargestC}{\Clock}(\symbstate_k))\in\Passed_k$ then, by induction on $i\leq p$, we have $\pval$ in 
		$\RIEF(\A,\symbstate_i,\LocsTarget,\Passed_i)$:
		\begin{itemize}
			\item if $i = p$ then, noting $\symbstate_i=(\loc_i,\C_i)$, we have  $\loc_i\in\LocsTarget$ and by \cref{cor:next}, $\pval\in\projectP{\C_i}$ and thus in the result of \RIEF{} as expected.
			\item otherwise, for $0<i\leq p$, since  we do not have $s_{i-1}\in\Passed_{i-1}$, we must be in the recursive case, and thus $\RIEF(\A,\symbstate_{i-1},\LocsTarget,\Passed_{i-1})$ contains $\RIEF(\A,\symbstate_i,\LocsTarget,\Passed_i)$, as part of the big union on recursive calls, and by the induction hypothesis it contains $\pval$ as well.
		\end{itemize}
		Then, in particular, for $i=0$, we have $\pval\in \RIEF(\A,\Init(\A),\LocsTarget,\emptyset)$.

		Assume now there exists $k$ such that $\IH(\Ext{\LargestC}{\Clock}(\symbstate_k))\in\Passed_k$ and let $j<k$ be such that $\IH(\Ext{\LargestC}{\Clock}(\symbstate_j))=\IH(\Ext{\LargestC}{\Clock}(\symbstate_k))$.

		Since $\pval$ is an integer parameter valuation, by
		\cref{lemma:extintv2}, this means that
		$\valuate{\Ext{\LargestC}{\Clock}(\symbstate_k)}{\pval} = \valuate{\Ext{\LargestC}{\Clock}(\symbstate_j)}{\pval}$.
		Using now \cref{lemma:extbisimeq}, for any $(\loc_k,\clockval_k)\in\valuate{\symbstate_k}{\pval}$ there exists a state
		$(\loc_k, \clockval_j)\in\valuate{\symbstate_j}{\pval}$ that is bisimilar to it.
		
		Also, by \cref{cor:next}, $(\loc_k, \clockval_j)$ is reachable in
		$\valuate{\A}{\pval}$ via some run~$\varrun_1$ along edges
		$\edge_1\ldots \edge_j$.  Also, since
		$(\loc_k, \clockval_j)$ and $(\loc_k, \clockval_k)$ are bisimilar, there
		exists a run~$\varrun_2$ that takes the same edges as the suffix of
		$\varrun$ starting at $(\loc_k, \clockval_k)$.
		Let $\varrun'$ be the run
		obtained by merging $\varrun_1$ and $\varrun_2$ at $(\loc_k, \clockval_j)$.
		Run~$\varrun'$ has strictly less discrete actions than $\varrun$ and also
		reaches $\LocsTarget$. We can thus repeat the same reasoning as we
		have just done.  We can do this only a finite number of times
		(because the length of the considered run is strictly decreasing)
		so at some point we have to be in some of the other cases and we
		obtain the expected result.
		\qedhere
	\end{enumerate}
\end{proof}
\begin{example}

    Consider the simple PTA with a unique transition from the initial location $\loc_0$ to~$\loc_1$ with guard $1\leq x\leq 2\param$.
    To ensure the $EF\{\loc_1\}$ property, we just need to be able to go through the transition from~$\loc_0$ to~$\loc_1$. 
	To compute polyhedron~$C_1$ obtained in~$\loc_1$, we first delay in~$\loc_0$, obtaining $x\geq 0$. Then we intersect with the guard, obtaining $1\leq x\leq 2\param$ (note that this implies $1\leq 2\param$).
	By delaying again in~$\loc_1$, we finally obtain that $C_1$ is $1\leq x \wedge 1 \leq 2\param$, which implies $\param\geq \frac{1}{2}$.
	The integer hull of~$C_1$ is $1\leq x \wedge 1\leq \param$, which implies $\param\geq 1$.

	Algorithm $\IEF$ gives the result $\param\geq 1 \land \param \in \setN$, while algorithm $\RIEF$ gives (on this example) the exact result $a\geq \frac{1}{2}$.
    \label{ex:RIEF}
\end{example}
\subsection{Parametric unavoidability: \RIAF}\label{ss:RIAF}

\begin{algorithm}[tb]
	\SetKwInOut{Input}{input}
	\SetKwInOut{Output}{output}

	\Input{A bounded PTA $\A$, a symbolic state $\symbstate = (\loc, \C)$, a set of target locations~$\LocsTarget$, a set~$\Passed$ of passed states on the current path}
	\Output{Set~$\K$ of parameter valuations guaranteeing unavoidability}

	\BlankLine

	\lIf{$\loc \in \LocsTarget$}{%
		$\K\assign \projectP{\C}$
	}\Else{%
        \lIf{$\Big(\loc,$\hl{$\IH$}$\big($\hl{$\Ext{\LargestC}{\Clock}$}$(\C) \big) \Big)\in \Passed$}{%
			$\K \assign \emptyset$\nllabel{algo:RIAF:passed}
		}\Else{%
			$\K\assign \projectP{\C} $\ ; \  $\KLive \assign \emptyset $\;
			\ForEach{outgoing $\edge = (\loc,\guard,\action,\resets,\loc')$ from $\loc$ in $\A$ s.t.\ $\Succ(\symbstate, \edge) \neq \emptyset$}{%
				$\symbstate' \assign \Succ(\symbstate, \edge)$\;
				
                $\KGood{} \assign \RIAF\Big(\A, \symbstate', \LocsTarget, \Passed \cup \big\{\big(\loc,$\hl{$\IH$}$($\hl{$\Ext{\LargestC}{\Clock}$}$(\C))\big) \big\}\Big)$\;
				$\KBlock{}\assign \KTrue \setminus \projectP{\symbstate'}$\;
				$\K\assign \K\cap (\KGood{} \cup \KBlock{})$\;\nllabel{algo:RIAF:union}
				$\KLive \assign \KLive  \cup \timepast{(\C\cap \valuations{\guard})}$\;
			}

			$\K\assign \K \setminus \projectP{(\C \setminus \KLive )}$\;
		}
	}

	\Return{$\K$}
\caption{$\RIAF(\A, (\loc, \C), \LocsTarget, \Passed)$}\label{algo:RIAF} %
\end{algorithm}

The goal of \RIAF{} (given in \cref{algo:RIAF}) is to synthesize parameter valuations solution to the AF-synthesis problem.
It is inspired by the algorithms \AF{} and \IAF{} introduced in~\cite{JLR15}; however \AF{} may not terminate, and \IAF{} can only output integer valuations.
Note also, as shown in \cref{ex:RIAF} below, that interpreting the result of \IAF{} as a dense set is incorrect in general, since it may contain non-integer valuations that do not ensure unavoidability.
The differences with \IAF{} (which was defined in~\cite{JLR15}) are highlighted in \cref{algo:RIAF}.

\RIAF{} works as a post-order traversal of the symbolic reachability tree, keeping valuations that permit to go into branches reaching $\LocsTarget$ and cutting off branches leading to a deadlock or looping without any occurrence of~$\LocsTarget$.
More precisely, \RIAF{} uses three sets of valuations:
\begin{enumerate}%
	\item \KGood{} contains the set of parameter valuations that indeed satisfy AF, recursively computed by calling \RIAF{};
	\item \KBlock{} allows to cut off branches, for instance leading to deadlock or looping, by keeping only parameter valuations in the complement of the first state in that branch: imagine we have a guard $2\leq \clock\leq \param$, but selecting $\param<2$ we effectively cut the branch off. Now, for each branch we need to ensure that it always either goes to $\LocsTarget$ (giving the constraint \KGood), or we cannot take it (constraint \KBlock), hence the union in the algorithm (\cref{algo:RIAF:union}).
	Note that for completeness we need also parameters cutting some branches that do go to \LocsTarget, provided they still allow at least one branch to go there: for instance in Figure~\ref{fig:IAF-counterex}, assuming $\LocsTarget=\{\loc_1, \loc_2\}$, all parameters valuations greater than $\frac{1}{2}$ ensure unavoidability since both branches go to \LocsTarget. But all valuations less or equal to $\frac{1}{2}$ also do because only the upper branch remains.  Parameters that cut all branches are forbidden using \KLive.
	\item $\KLive$ (which is a set of valuations over $\Clock \cup \Param$) is necessary to forbid reaching states from which no transition can be taken for any $\edge$, even after some delay, and also to forbid cutting all branches.
\end{enumerate}

The main difference between \AF{} and \RIAF{} is that we use the convergence condition of \IAF{}, which operates on integer hulls instead of symbolic states, hence ensuring termination with the same reasoning as~\RIEF{} (\cref{lemma:finite}).

We state below the soundness and integer completeness of~\RIAF{}.
The proofs are similar to those of $\AF$ in~\cite{JLR15}, by using the additional arguments provided in the proof of~$\RIEF$, and in particular \cref{lemma:finite}.

\begin{theorem}\label{theorem:RIAF:soundness}
	Let $\K = \RIAF(\A,\Init(\A), \LocsTarget, \emptyset)$.
    We have:
    \begin{enumerate}
        \item Soundness: If $\pval \in \K$ then $\LocsTarget$ is unavoidable in $\valuate{\A}{\pval}$;
        \item Integer completeness: If $\pval$ is an integer parameter valuation, and $\LocsTarget$ is unavoidable in $\valuate{\A}{\pval}$ then $\pval \in \K$.
    \end{enumerate}
\end{theorem}
\begin{proof}
		The algorithm having terminated, it has explored a finite prefix $\treeprefix$ of
		$\symtree$. 
\begin{enumerate}
	\item Soundness: 
	Let $\symbstate$ be symbolic state in $\treeprefix$, and $\Passed$ be the set of ancestors of $\symbstate$ in $\treeprefix$. Further let $\Passed'$ be the set $\{(\loc,\IH(\Ext{\LargestC}{\Clock}(\C))\,\mid\, (\loc,\C)\in \Passed\}$.

	Then, by induction, we have $\RIAF(\A, \symbstate, \LocsTarget, \Passed')\subseteq \AF(\A, \symbstate, \LocsTarget, \Passed)$:
	\begin{itemize}
		\item if $\symbstate=(\loc,\C)$ is a leaf of $\treeprefix$, then either $\loc\in\LocsTarget$ and both algorithms return the same value $\projectP{\C}$, or the result of \RIAF{} is empty and thus included in whatever is the result of~\AF{}.
		\item if it is not a leaf, then by the induction hypothesis we have inclusion for all the constraints $\Kgood{}$ from the recursive calls. Polyhedra $\KLive{}$ and $\KBlock{}$ have the same value in both algorithms, and so it follows that the $\K$ returned by $\RIAF$ is included in the one returned by $\AF$.
    \end{itemize}

	\item Integer completeness: 
		Let $\pval$ be an \emph{integer} parameter valuation for which the unavoidability property holds in $\valuate{\A}{\pval}$. 

		Given a symbolic state~$\symbstate$ in~$\treeprefix$, let $\Passed$ be the set of all $\IH(\Ext{\LargestC}{\Clock}(\symbstate'))$ for all ancestors $\symbstate'$ of $\symbstate$ in~$\treeprefix$.
		We prove by induction on the finite tree~$\treeprefix$ that, for each symbolic state $\symbstate$ in $\treeprefix$ such that the sequence of edges leading to~$\symbstate$ is feasible in~$\valuate{\A}{\pval}$,
		  we have $\pval\in \RIAF(\A,\symbstate, \LocsTarget, \Passed)$:
		\begin{itemize}
        	\item if $\symbstate=(\loc,\C)$ is a leaf then  either it is because $\loc\in\LocsTarget$. Then, by \cref{cor:next}, $\pval$ is in $\projectP{\C}$ and thus in the result of \RIAF{} as expected.

			Or it is because there exists an ancestor $\symbstate'$ of $\symbstate$ such that 
			$\IH(\Ext{\LargestC}{\Clock}(\symbstate')) = \IH(\Ext{\LargestC}{\Clock}(\symbstate))$.			
		Since $\pval$ is an integer parameter valuation, by
		\cref{lemma:extintv2}, this means that
		$\valuate{\Ext{\LargestC}{\Clock}(\symbstate')}{\pval} = \valuate{\Ext{\LargestC}{\Clock}(\symbstate)}{\pval}$.
		
		Using \cref{lemma:extsucc}, we can replay the sequence of actions from $\symbstate'$ to $\symbstate$ from $\symbstate$ and obtain a new symbolic state $\symbstate''$ such that $\valuate{\Ext{\LargestC}{\Clock}(\symbstate'')}{\pval} = \valuate{\Ext{\LargestC}{\Clock}(\symbstate)}{\pval}$. This implies that we can repeat this procedure and obtain a path of symbolic states of $\valuate{\A}{\pval}$, and a corresponding concrete path using \cref{cor:next}, longer than any given length. Those paths never go through $\LocsTarget$, because by construction none of the ancestors of $\symbstate$ has a location in $\LocsTarget$, otherwise that node would have no successor in~$\treeprefix$.

		Finally, it follows from \cite{alur-IC-93} that since the unavoidability property holds, we must be able to reach $\LocsTarget$ within some fixed number of discrete transitions that depends only on the number of places, transitions, clocks and the maximal constant in $\valuate{\A}{\pval}$\footnote{Actually this would be the number of states in the \emph{region graph} of $\valuate{\A}{\pval}$ as defined in~\cite{alur-IC-93}, since a loop in that graph would imply an infinite run in $\valuate{\A}{\pval}$.}.

	It is therefore not possible that there exists an ancestor $\symbstate'$ of~$\symbstate$ such that
			$\IH(\Ext{\LargestC}{\Clock}(\symbstate')) = \IH(\Ext{\LargestC}{\Clock}(\symbstate))$.

		Finally, it may be that none of the two previous cases hold but $\symbstate$ has no successor. But then, by \cref{cor:next}, any state $(\loc, \clockval)\in\symbstate$ is reachable in
		$\valuate{\A}{\pval}$ via some run~$\varrun$ that, as before, never goes through $\LocsTarget$ and has no discrete successor, even after some delay. As before this is not possible.
			
			\item if $\symbstate=(\loc,\C)$ is not a leaf then we must be in the recursive case and there must exist at least one successor of $\symbstate$.

			Consider successor $\symbstate'$ of $\symbstate$ by edge $\edge$. There are two possibilities:
			\begin{itemize}
				\item either $\edge$ is feasible in $\valuate{\A}{\pval}$ from some state $\state$ in $\valuate{\symbstate}{\pval}$; we know there is such a state because we have assumed that the sequence of edges leading to $\symbstate$ is feasible in $\valuate{\A}{\pval}$ and thus, by \cref{cor:next}, $\valuate{\C}{\pval}$ is not empty.

				Then by the induction hypothesis, $\pval$ is in $\KGood$. We also have that the guard of $\edge$ is satisfiable from $\state$, possibly after some delay. This implies that $\pval$ is also in $\KLive$ and will be at the end of the \texttt{foreach} loop.
				\item or $\edge$ is not feasible in $\valuate{\A}{\pval}$ from any state in $\valuate{\symbstate}{\pval}$. Then by \cref{cor:next}, $\pval$ is not in~$\projectP{\symbstate'}$, and thus in~$\KBlock$.
			\end{itemize}

			The first case must happen at least once, otherwise by \cref{cor:next} we have a deadlocked run in $\valuate{\A}{\pval}$ never reaching $\LocsTarget$. Therefore we have $\pval\in\KLive$ at the end of the \texttt{foreach} loop, and for all successors we have $\pval$ either in $\KGood$ or in $\KBlock$.
			It follows that $\pval$ is in the returned value of $\K$. 
		\qedhere
		\end{itemize}
		
	\end{enumerate}
\end{proof}
\begin{example}
	\begin{figure}
		\centering
		\begin{tikzpicture}[PTA, node distance=3.5cm, thin]

			\node[location, initial] (l0) {$\loc_0$};
			\node[location, right of=l0] (l1) {$\loc_1$};
			\node[location, below of=l1, yshift=2.5cm] (l2) {$\loc_2$};
			
			\path
				(l0) edge[] node {$\clockx \geq 0$} (l1)
				(l0) edge[bend right] node[below left] {$1 \leq \clockx \leq 2 \paramp$} (l2)
				;
		\end{tikzpicture}
        \caption{Counter-example to the density of the result of $\IAF$.} %
		\label{fig:IAF-counterex}
	\end{figure}
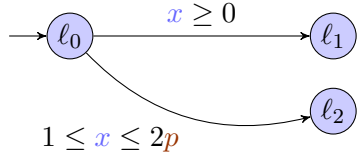

	Consider the PTA in \cref{fig:IAF-counterex}. To ensure the $ AF \{\loc_1\}$ property, we need to cut the transition from~$\loc_0$ to~$\loc_2$.
	With a computation similar to \cref{ex:RIEF}, the polyhedron $C_2$ obtained in $\loc_2$ is $1\leq x \wedge 1 \leq 2\param$, which implies $\param\geq \frac{1}{2}$.
	The integer hull of $C_2$ is $1\leq x \wedge 1\leq \param$, which implies $\param\geq 1$.
	In order to block the path to~$\loc_2$ in~$\loc_0$, we intersect with the complement of the projection onto the parameters of $\IH(C_2)$, \ie{} $\param<1$.
	Since there is no constraint on the transition from~$\loc_0$ to~$\loc_1$ the final result of the former algorithm $\IAF$ from~\cite{JLR15} is actually $\param<1$.
	For integer parameters this means $\param=0$, which is correct. But if we interpret the result over rational-valued parameters, we obtain that, for instance, $\param=\frac{1}{2}$ should be a valuation ensuring the property---while it is obviously not.

	On the same example, $\RIAF$ gives (on this example) the exact result, that is $\param<\frac{1}{2}$.
    \label{ex:RIAF}
\end{example}
\subsection{Parametric traces preservation: \RITPS{}}\label{ss:IM}

We address here the traces preservation synthesis problem.
That is, given a PTA~$\A$ and a reference parameter valuation~$\pval$, we seek other valuations~$\pval'$ such that $\valuate{\A}{\pval'}$ has the same trace set as~$\valuate{\A}{\pval}$.
In~\cite{ACEF09,ALM20}, we defined a procedure \TPS{} (for ``traces preservation  synthesis'', and also called ``inverse method'') that may not terminate, but is correct whenever it terminates.
$\TPS$ is implemented in the \imitator{} software~\cite{Andre21}.

\subsubsection{The algorithm}\label{ss:IM:algorithm}

We give $\RITPS$ in \cref{algo:RITPSrecursive}.
The differences with the original \TPS{} (formalized in~\cite{ALM20}) are highlighted.

$\RITPS$ is a recursive algorithm exploring the state space in a depth-first search manner.
Given a current symbolic state, its $\pval$-compatibility is first checked~(\cref{algo:RITPS:check}).
\begin{itemize}
	\item If this state is $\pval$-compatible, the current constraint is intersected with the projection onto the parameters of \hl{the integer hull of} this state (\cref{algo:RITPSrecursive:interIH}).
	Then, if this state was not visited up to the integer hull (\ie{} if the integer hull of this state does not belong to the list of the integer hulls of passed states, \cref{algo:RITPSrecursive:checkpassed}), then all its successors are computed, and the results of the recursive calls are intersected with the current constraint.
	Note that \hl{the integer hull of} the current state is added to the list of passed states in the recursive calls---which is the crux for ensuring termination.

	\item Conversely, if the current state is not $\pval$-compatible, \hl{the integer hull of} the negation of its projection onto the parameters is intersected with the current constraint (\cref{algo:RITPSrecursive:piincompatible}).
\end{itemize}

\begin{algorithm}[t]
	\SetKwInOut{Input}{input}
	\SetKwInOut{Output}{output}

	\Input{A bounded PTA~$\A$, an \hl{integer} parameter valuation $\pval$, a current symbolic state $\symbstate = (\loc, \C)$, a set $\Passed$ of passed states on the current path} %
	\Output{Set $\K$ of parameter valuations guaranteeing traces preservation}

	\BlankLine
	$
		\K \assign \KTrue
	$

	\eIf{$\pval \in \projectP{\C}$\nllabel{algo:RITPS:check}}{%
		$\K \assign \K \cap \projectP{\Big($\hl{$\IH$}$\big($\hl{$\Ext{\LargestC}{\Clock}$}$(\C) \big)\Big)}$\label{algo:RITPSrecursive:interIH}

		\If{$ $\hl{$\IH$}$\big($\hl{$\Ext{\LargestC}{\Clock}$}$(\symbstate)\big) \notin \Passed$\nllabel{algo:RITPSrecursive:checkpassed}}{%

			\ForEach{outgoing $\edge$ from $\loc$ in $\A$ s.t.\ $\Succ(\symbstate, \edge) \neq \emptyset$\nllabel{algo:RITPS:foreach}}{
				$\symbstate' \assign \Succ(\symbstate, \edge)$\nllabel{algo:RITPS:succ}

				$\K \assign \K \cap \RITPS\Big(\A, \pval, \symbstate', \Passed \cup \big\{ $\hl{$\IH$}$($\hl{$\Ext{\LargestC}{\Clock}$}$(\symbstate)) \big\} \Big) $\nllabel{algo:RITPS:rec}
			} %
		}
	}
	{
		$\K \assign \K \cap \, $\hl{$\IH$}$\big(\neg (\projectP{\C}) \big)$\nllabel{algo:RITPSrecursive:piincompatible}
	}

	\Return{$\K$} \nllabel{algo:RITPS:return}

	\caption{$\RITPS(\A, \pval, \symbstate, \Passed)$}
	\label{algo:RITPSrecursive}
\end{algorithm}

The base call is $\RITPS(\A, \pval, \Init(\A), \emptyset)$.

The algorithm~$\RITPS$ differs from~$\TPS$ in four main aspects.
\begin{enumerate}
	\item The termination condition is not anymore the equality of the states, but of the \emph{integer hull} of the extrapolation of the states (\cref{algo:RITPSrecursive:checkpassed}): this is needed to ensure termination.
	
	\item $\RITPS$ does not return the intersection of the $\pval$-compatible states' sets of parameter valuations, but returns the intersection of their \emph{integer hulls} (\cref{algo:RITPSrecursive:interIH}):
	this is needed to ensure soundness (see a counter-example in \cref{counterexample:RITPS:IH}).
	
	\item $\RITPS$ does not return the intersection of the negation of the $\pval$-incompatible states' sets of parameter valuations, but returns the intersection of the \emph{integer hulls} of their negation (\cref{algo:RITPSrecursive:piincompatible}):
	this is needed to ensure soundness.

	\item $\pval$ must be an \emph{integer valuation}:
	this is again needed to ensure soundness (see a counter-example in \cref{counterexample:RITPS:integer}).
\end{enumerate}
We believe that requiring $\pval$ to be an integer is harmless in practice, since a rational can be turned into an integer by appropriately resizing the constants of the PTA.
However, the result of~$\RITPS$ still outputs in general a \emph{dense} set of rationals (just as~$\TPS$), and hence can synthesise parameter valuations infinitesimally close to~$\pval$ with the same set of traces---producing a (potentially partial) measure of the system robustness with respect to the set of traces (see~\cite{APP13}).

\subsubsection[Examples]{Examples}
\begin{example}\label{example:RITPS:termine}
	\begin{figure}[tb]
	{
	\centering
	
\begin{subfigure}[b]{0.35\textwidth}
	\begin{tikzpicture}[PTA, thin, node distance=4cm]

		\node[location, initial] (l0) {$\loc_1$};
		\node[location, right of=l0] (l1) {$\loc_1'$};

		\node [invariant, above] at (l0.north) {$\clocki{1} \leq \parami{2}$};
		\node [invariant, above] at (l1.north) {$\clocki{2} \leq \parami{1}$};

		\path
			(l0) edge[bend left] node[above,align=center] {$\styleact{a}$\\$\clocki{1} \assign 0$} (l1)
			(l1) edge[bend left] node[above,yshift=-2,align=center] {$\styleact{b}$\\$\clocki{2} = \parami{1}$} node[below]{$\clocki{2} \assign 0$} (l0)
			;
	\end{tikzpicture}
	\caption{PTA}
	\label{figure:RITPS:termine}
\end{subfigure}
\hfill
\begin{subfigure}[b]{0.55\textwidth}
		\begin{tikzpicture}[scale=.7]
			\small
			\tikzstyle{axe} = [line width=1pt, ->, draw=black!80]
			\tikzstyle{fondgris} = [fill=black!5, draw=none]
			\tikzstyle{C1} = [line width=1pt, draw=green!50!black, dashed]
			\tikzstyle{C2} = [line width=1pt, draw=red!50!black, dotted]
			\tikzstyle{contrainte} = [fill=blue!50!white, draw=blue!30!black]

			\draw[fondgris] (0, 0) rectangle (5, 5);

			\draw[contrainte] (0, 0) -- (5, 5) -- (0, 5) -- cycle;

			\draw[C1] (0, 0) -- (5, 5*6/7) -- (5, 5) -- (0, 5) -- cycle;
			\draw[C2] (0, 0) -- (5, 5*13/14) -- (5, 5) -- (0, 5) -- cycle;

			\draw[C2] (7, 4) --++ (1, 0);
			\node at (8, 4) [right] {$7 \parami{2} \geq 6 \parami{1}$};
			\draw[C1] (7, 3) --++ (1, 0);
			\node at (8, 3) [right] {$6 \parami{2} \geq 5 \parami{1}$};
			
			\node at (2, 4.5) {$\RITPS(\A, \pval, \emptyset)$};

			\path[axe] %
				(0, 0) -- ++ (6.5, 0);
			\node at (7, -.5) {$\parami{1}$};
			\path[axe] %
				(0, 0) -- ++ (0, 5.5);
			\node at (-1.2, 5.2) {$\parami{2}$};

			\foreach \x in {0, 1, ..., 6} %
				\foreach \y in {0, 1, ..., 5} %
					\node at (\x, \y) {\Large{$\cdot$}};

			\node at (1, 2) {\Huge{$\textcolor{red!50!black}{\cdot}$}};
			\node at (1, 2.5) {$\textcolor{red!50!black}{\pval}$};

			\foreach \x in {0, 1, ..., 5} %
				\node at (-.5, \x) {$\x$};
			\foreach \x in {0, 1, ..., 5} %
				\node at (\x, -.5) {$\x$};

		\end{tikzpicture}
	\caption{Sets of parameter valuations}
	\label{figure:ex:IM:IH}
\end{subfigure}

	}
	\caption{Example for which $\RITPS$ terminates but not $\TPS$}
	\end{figure}
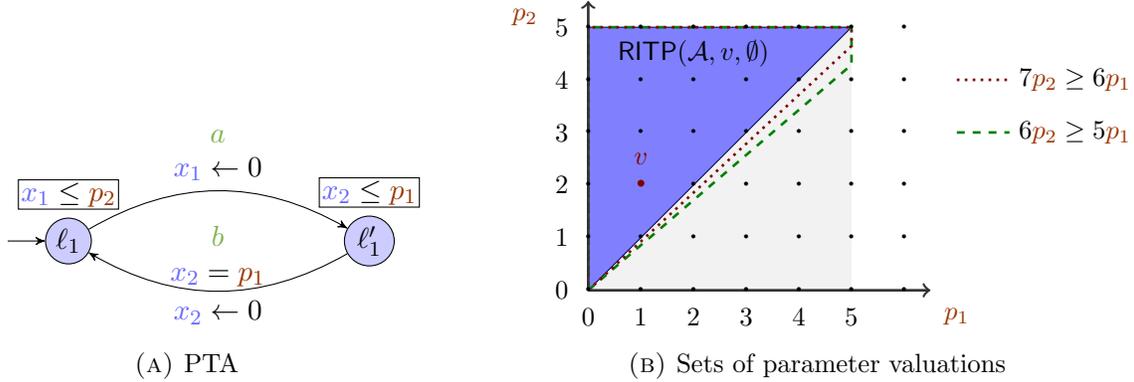

	Consider the PTA~$\A$ depicted in \cref{figure:RITPS:termine}; assume that the domain of~$\paramitext{1}$ and~$\paramitext{2}$ is $[0 , 5]$.
	Let $\pval$ be the valuation such that $\pval(\paramitext{1}) = 1$ and $\pval(\paramitext{2}) = 2$.
	For this valuation, the trace set is made of the single following infinite trace alternating between the two locations:
	\[\loc_1  \FlecheConcrete{a} \loc_1' \FlecheConcrete {b} \loc_1  \FlecheConcrete{a} \loc_1' \FlecheConcrete {b} \cdots\]

	$\TPS(\A, \pval, \Init(\A), \emptyset)$ does not terminate, because sets of parameter valuations are generated of the form
		$\clockitext{1} - \clockitext{2} \geq i \times (\paramitext{1} - \paramitext{2})$, with $i \in \setN$ growing without bound,
	and all these sets of parameter valuations are incomparable.
	In contrast, for $\RITPS$, the integer hull of the extrapolation of these polyhedra will eventually be the same (in fact, no extrapolation is needed, because both clocks are always $\leq 5$ due to the invariants and the bounded parameter domain); in \cref{figure:ex:IM:IH}, we give two such sets of parameter valuations, \ie{}
		$6 \paramitext{2} \geq 5 \paramitext{1}$
		and
		$7 \paramitext{2} \geq 6 \paramitext{1}$, that indeed contain the same set of integer points within $[0 , 5]^2$, and hence the integer hull (of their extrapolation) is equal.%
	\footnote{%
		Strictly speaking, the sets of parameter valuations are over $\Clock \cup \Param$, \ie{} in 4 dimensions (2~clocks and 2~parameters); for the sake of better understanding, in \cref{figure:ex:IM:IH}, we project them onto~$\Param$ so as to represent them in 2~dimensions.
	}
	Eventually, $\RITPS$ returns
		$0 \leq \paramitext{1} \leq \paramitext{2} \leq 5$.
	In fact, for this example, it can be shown that the result output by $\RITPS$ is the maximal result, \ie{} any valuation such that $\paramitext{1} > \paramitext{2}$ has a trace set different from~$\pval$ (which will be a single trace made of a \emph{finite} alternation between $\loc_1$ and~$\loc_1'$, with the trace length proportional to the ratio between~$\paramitext{1}$ and~$\paramitext{2}$).
\end{example}

The following example illustrates the necessity to return the \emph{integer hull of} the $\pval$-compatible states (\cref{algo:RITPSrecursive:interIH}); otherwise, $\RITPS$ would output an incorrect result.

	\begin{figure}[tb]
	{
	\centering

	\begin{tikzpicture}[PTA, thin, xscale=1.4]

		\node[location, initial] %
	at(1.5,0) (l1) {$\loc_1$};
		\node[location] at(3.5,0) (l2) {$\loc_2$};
		\node[location] at(6,0) (l3) {$\loc_3$};

		\node [invariant, above] at (l1.north) {$\clockx \leq 2\paramp$};

		\path
			(l1) edge[] node [above,align=center] {$\clockx = 1$} node[below]{$\clockx \assign 0$} (l2)
			(l2) edge[loop above] node[align=center] {$\clockx = 1 \land \clockx \leq \paramp$\\$\clockx \assign 0$} (l2)
			(l2) edge[] node [above,align=center] {$\paramp = \clockx = \paramq - 1$} (l3)
			;
	\end{tikzpicture}

	}
	\caption{Example illustrating the need of $\IH$ in the result of $\RITPS$}
	\label{figure:RITPS:IHbis}
	\end{figure}
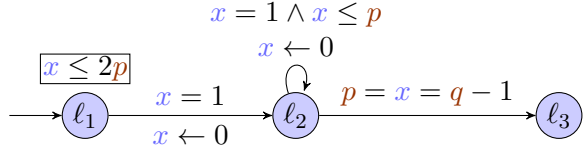
\begin{example}\label{counterexample:RITPS:IH}
	Consider the PTA~$\A$ depicted in \cref{figure:RITPS:IHbis}, featuring a single clock~$\clockxtext$ and two parameters (actions are not depicted, because not relevant for this example---one can assume that every edge is labelled with some action~$a$); assume that the bounded parameter domain of~$\paramptext$ and $\paramqtext$ is $[0;2]$.
	Let $\pval$ be such that $\pval(\paramptext) = 1$ and $\pval(\paramqtext) = 2$.
	Then $\TPS(\A\, \pval, \Init(\A), \emptyset)$ synthesizes the set of valuations~$\K$ restricted to~$\K = \{ \pval \}$.

	Now, consider the algorithm $\RITPS$.
	The first time $\loc_2$ is visited, the associated set of valuations $\C_2$ is:
	$\frac{1}{2} \leq \paramptext \leq 2$ (other trivial inequalities such as $\clockxtext \geq 0$ are omitted; also note that the upper bound for~$\paramptext$ comes from its bounded parameter domain).
	Its successor is $(\loc_3, \C_3)$ with $\C_3$ being:
	$\frac{1}{2} \leq \paramptext \leq 1 \land \paramptext \leq \clockxtext \land \paramqtext = \paramptext + 1$.
	The projection of this set onto~$\Param$ is
	$\paramqtext = \paramptext + 1 \land \frac{1}{2} \leq \paramptext \leq 1$. %

	When visiting $\loc_2$ for the second time (via the self-loop over~$\loc_2$), the associated set of valuations $\C_2'$ is:
	$1 \leq \paramptext \leq 2$.
	Observe that $\IH(\C_2) = \IH(\C_2')$.
	Therefore, $\RITPS$ does not explore the successors of~$\C_2'$;
	if $\RITPS$ returned $\projectP{\C}$ at \cref{algo:RITPSrecursive:interIH}, then the result would be $\paramqtext = \paramptext + 1 \land \frac{1}{2} \leq \paramptext \leq 1$---which is wrong.

	Now, $\RITPS$ actually returns $\projectP{\Big(\IH\big(\Ext{\LargestC}{\Clock}(\C_3) \big)\Big)}$, which gives the expected result $\paramptext = 1 \land \paramqtext = 2$.
\end{example}
\begin{figure}[tb]
{
\centering

\begin{tikzpicture}[PTA, thin, xscale=1]

	\node[location, initial] at(0, 0) (l1) {$\loc_1$};
	\node[location] at(3, 0) (l2) {$\loc_2$};

	\node [invariant, above] at (l1.north) {$\clockx \leq 2\paramp$};

	\path
		(l1) edge[] node [above,align=center] {$\clockx = 1$} (l2);
\end{tikzpicture}

}
\caption{Example illustrating the need for an integer valuation in $\RITPS$}
\label{figure:RITPS:counterexample-integer}
\end{figure}
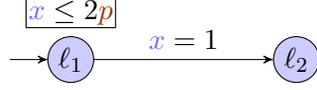

The following example illustrates the necessity of requiring the input $\pval$ of~$\RITPS$ to be an integer.

\begin{example}\label{counterexample:RITPS:integer}
	Consider the PTA~$\A$ in \cref{figure:RITPS:counterexample-integer}.
	Consider $\pval$ such that $\pval(\paramptext) = \frac{1}{2}$.
	Then $\RITPS(\A, \pval, \Init(\A), \emptyset)$ outputs the set~$\K$ of parameter valuations defined by $1 \leq \paramptext \leq 2$.
	Then obviously $\pval \notin \K$, which contradicts the expected soundness (see \cref{theorem:RITPS:soundness}).
\end{example}
\subsubsection{Termination and soundness}

Termination is immediate from the finiteness of the integer hulls of extrapolations (\cref{lemma:finite}), which makes \RITPS{} only explore a finite number of states.

We state below the soundness and integer completeness of~\RITPS{}.

We first need the following lemma from~\cite{JLR15}.

\begin{lemC}[{\cite[Lemma~6]{JLR15}}]\label{lemma:JLR15:lemma6}
	Given a symbolic state $\symbstate$ and an edge~$\edge$, we have
	$\IH(\Succ(\IH(\symbstate), \edge)) = \IH(\Succ(\symbstate, \edge))$.
\end{lemC}

The following result derives immediately:

\begin{lemma}\label{lemma:IH-Succ}
	Given two symbolic states $\symbstate_1$ and~$\symbstate_2$ such that $\IH(\symbstate_1) = \IH(\symbstate_2)$, and an edge~$\edge$, we have
	$\IH(\Succ(\symbstate_1, \edge)) = \IH(\Succ(\symbstate_2, \edge))$.
\end{lemma}
\begin{proof}
	$\IH(\Succ(\symbstate_1, \edge)) = \IH(\Succ(\IH(\symbstate_1), \edge))$ (by \cref{lemma:JLR15:lemma6}),
	then because $\IH(\symbstate_1) = \IH(\symbstate_2)$, we have
	$\IH(\Succ(\symbstate_1, \edge)) = \IH(\Succ(\IH(\symbstate_2), \edge)) = \IH(\Succ(\symbstate_2, \edge))$.
\end{proof}

The following results are immediate, yet allow our proof of \cref{proposition:RITPS:soundness} to be more clear.

\begin{lemma}\label{lemma:IHExt:projection}
	Given two symbolic states $\symbstate_1$ and~$\symbstate_2$ such that $\IH \big(\Ext{\LargestC}{\Clock}(\symbstate_1)\big) = \IH\big(\Ext{\LargestC}{\Clock}(\symbstate_2)\big)$, we have
	$\projectP{\Big(\IH\big(\Ext{\LargestC}{\Clock}(\symbstate_1)\big)\Big)} = \projectP{\Big(\IH\big(\Ext{\LargestC}{\Clock}(\symbstate_2)\big)\Big)}$.
\end{lemma}

\begin{lemma}\label{lemma:IHExt:projection-inside}
	Given two symbolic states $\symbstate_1$ and~$\symbstate_2$ such that $\IH \big(\Ext{\LargestC}{\Clock}(\symbstate_1)\big) = \IH\big(\Ext{\LargestC}{\Clock}(\symbstate_2)\big)$, we have
	$\IH\Big(\projectP{\Big(\IH\big(\Ext{\LargestC}{\Clock}(\symbstate_1)\big)\Big)}\Big) = \IH\Big(\projectP{\Big(\IH\big(\Ext{\LargestC}{\Clock}(\symbstate_2)\big)\Big)}\Big)$.
\end{lemma}

\begin{proposition}[Soundness of~$\RITPS$]\label{proposition:RITPS:soundness}
    Upon termination of $\RITPS$, if $\pval' \in \RITPS(\A, \pval, \Init(\A), \emptyset)$ then
			$\Traces(\valuate{\A}{\pval'}) = \Traces(\valuate{\A}{\pval})$.
\end{proposition}
\begin{proof}
Since \RITPS{} has terminated, it has explored a finite prefix $\treeprefix$ of the symbolic reachability tree~$\symtree$.
	Let $\symbstate$ be symbolic state in $\treeprefix$, and $\Passed$ be the set of ancestors of $\symbstate$ in $\treeprefix$.
	Further let $\Passed'$ be the set $\{(\loc,\IH(\Ext{\LargestC}{\Clock}(\C)) \mid (\loc,\C)\in \Passed\}$.

	Let us show by induction that we have $\RITPS(\A, \pval, \symbstate, \Passed') \subseteq \TPS(\A, \pval, \symbstate, \Passed)$:
	\begin{itemize}
		\item Base case: If $\symbstate = (\loc,\C)$ is a leaf of~$\treeprefix$, then
			\begin{itemize}
				\item either $\symbstate$ is $\pval$-compatible and $\RITPS$ returns $\projectP{(\IH(\Ext{\LargestC}{\Clock}(\C)))}$ while $\TPS$ returns $\projectP{\C}$; by \cref{lemma:extproj} and since $\IH(\C) \subseteq \C$ (by definition of~$\IH$), the inclusion thus holds;

				\item or $\symbstate$ is $\pval$-incompatible and $\RITPS$ returns $\IH(\neg \projectP{\C})$ while $\TPS$ returns $\neg\projectP{\C}$, and the inclusion holds again.
			\end{itemize}

		\item Induction case: If $\symbstate$ is not a leaf, then two cases arise.
		\begin{itemize}
			\item If $\IH(\Ext{\LargestC}{\Clock}(\symbstate)) \notin \Passed'$, by induction hypothesis, the inclusion holds again.

			\item The case when $\IH(\Ext{\LargestC}{\Clock}(\symbstate)) \in \Passed'$ is more delicate, as this situation did not occur in the proofs of \cref{theorem:RIEF:soundness,theorem:RIAF:soundness}.
			Since $\IH(\Ext{\LargestC}{\Clock}(\symbstate)) \in \Passed'$, then there exists $\symbstate'' = (\loc, \C'') \in \Passed'$ such that $\IH(\Ext{\LargestC}{\Clock}(\symbstate)) = \IH(\Ext{\LargestC}{\Clock}(\symbstate''))$, and $\symbstate''$ is currently being explored by~$\RITPS$, \ie{} a call to $\RITPS(\A, \pval, \symbstate'', \Passed'')$ was made before, and this call led recursively to the call to $\RITPS(\A, \pval, \symbstate, \Passed')$.
			Hence, $\Passed''$ denotes the set of passed states when $\RITPS$ was called on~$\symbstate''$.
			Let us now show by induction that whenever $\Passed'' \subseteq \Passed'$ and $\IH(\Ext{\LargestC}{\Clock}(\symbstate)) = \IH(\Ext{\LargestC}{\Clock}(\symbstate''))$ then $\RITPS(\A, \pval, \symbstate'', \Passed'') = \RITPS(\A, \pval, \symbstate, \Passed')$.

			If $\symbstate''$ is a leaf state and is $\pval$-compatible, then because $\pval$ is an integer valuation and because $\IH(\Ext{\LargestC}{\Clock}(\symbstate)) = \IH(\Ext{\LargestC}{\Clock}(\symbstate''))$, then $\symbstate$ is also $\pval$-compatible.
			For the same reason and by \cref{lemma:IHExt:projection}, $\projectP{(\IH(\Ext{\LargestC}{\Clock}(\symbstate)))} = \projectP{(\IH(\Ext{\LargestC}{\Clock}(\symbstate'')))}$ and both calls to~$\RITPS$ return the same result.

			If $\symbstate''$ is a leaf state and is $\pval$-incompatible, then for the same reason and by \cref{lemma:IHExt:projection-inside}, both calls to~$\RITPS$ return the same result.

			Now, let us consider the induction case.
			Again, if $\symbstate''$ is $\pval$-compatible, then because $\pval$ is an integer valuation and because $\IH(\Ext{\LargestC}{\Clock}(\symbstate)) = \IH(\Ext{\LargestC}{\Clock}(\symbstate''))$, then $\symbstate$ is also $\pval$-compatible, and the intersection \cref{algo:RITPSrecursive:interIH} gives the same result.
			If $\symbstate'' \in \Passed''$, then because $\IH(\Ext{\LargestC}{\Clock}(\symbstate)) = \IH(\Ext{\LargestC}{\Clock}(\symbstate''))$ and $\Passed'' \subseteq \Passed'$, we have that $\symbstate \in \Passed'$ as well, and therefore both calls return the same result.

			Now, if $\symbstate'' \notin \Passed''$, then for each edge~$\edge$, because $\IH(\Ext{\LargestC}{\Clock}(\symbstate)) = \IH(\Ext{\LargestC}{\Clock}(\symbstate''))$, we have that
			$\IH(\Ext{\LargestC}{\Clock}(\Succ(\symbstate , \edge))) = \IH(\Ext{\LargestC}{\Clock}(\Succ(\symbstate'', \edge)))$ by \cref{lemma:IH-Succ}.
			In addition, because $\Passed'' \subseteq \Passed'$ by induction, then we have $\Passed'' \cup \{ \IH(\Ext{\LargestC}{\Clock}(\symbstate'')) \} \subseteq \Passed' \cup \{ \IH(\Ext{\LargestC}{\Clock}(\symbstate)) \}$.
			Hence, by induction, both calls to~$\RITPS$ return the same result.
		\end{itemize}
		As a consequence, discarding a state because it is equal to a former visited state (\cref{algo:RITPSrecursive:checkpassed}) is harmless because it would give the same result.
		The visit of~$\RITPS$ to that former state explores (by definition) the same tree prefix as~$\TPS$, up to the integer hulls.
		Hence, both algorithms explore the same tree prefix (up to the integer hulls), and therefore by induction hypothesis, $\RITPS(\A, \pval, \symbstate, \Passed') \subseteq \TPS(\A, \pval, \symbstate, \Passed)$.

		Finally, since for all $\pval' \in \TPS(\A, \pval, \Init(\A), \emptyset), \Traces(\valuate{\A}{\pval'}) = \Traces(\valuate{\A}{\pval})$ from~\cite{ALM20}, then for all $\pval' \in \RITPS(\A, \pval, \Init(\A), \emptyset)$, it holds that $\Traces(\valuate{\A}{\pval'}) = \Traces(\valuate{\A}{\pval})$.
		\qedhere
    \end{itemize}
\end{proof}
\begin{proposition}[Integer completeness of~$\RITPS$]\label{proposition:RITPS:integercompleteness}
    Upon termination of $\RITPS$, if $\pval'$ is an integer parameter valuation, and $\Traces(\valuate{\A}{\pval'}) = \Traces(\valuate{\A}{\pval})$, then $\pval' \in \RITPS(\A, \pval, \Init(\A), \emptyset)$.
\end{proposition}
\begin{proof}
	Let us show that, for any integer parameter valuation~$\pval'$, $\pval' \in \RITPS(\A, \pval, \Init(\A), \emptyset)$ iff $\pval' \in \TPS(\A, \pval, \Init(\A), \emptyset)$.
	The proof essentially follows the induction case of the proof of \cref{proposition:RITPS:soundness}.
	That is, given a finite tree prefix of the symbolic reachability tree ending in a symbolic state~$\symbstate$, then either $\RITPS$ explores the same tree prefix, and both algorithms return a similar result---up to their integer hull---which ensures the fact that an \emph{integer} parameter valuation belongs to the result of~$\RITPS$ iff it belongs to the result of~$\TPS$.
	Or $\RITPS$ did not explore that prefix in full because, from a given symbolic state~$\symbstate$, $\RITPS$ found a former state $\symbstate'$ such that $\IH(\Ext{\LargestC}{\Clock}(\symbstate)) = \IH(\Ext{\LargestC}{\Clock}(\symbstate'))$ and stopped exploring.
	But, from the induction case in the proof of \cref{proposition:RITPS:soundness}, exploring the tree from~$\symbstate'$ would yield the same result as the exploration from~$\symbstate$ and would match the result of~$\TPS$ up to the integer hull, and therefore for any \emph{integer} parameter valuation, the results of~$\TPS$ and~$\RITPS$ coincide.

	The result finally follows from the fact that for any parameter valuation~$\pval'$, $\Traces(\valuate{\A}{\pval'}) = \Traces(\valuate{\A}{\pval})$ iff $\pval' \in \TPS(\A, \pval, \Init(\A), \emptyset)$~\cite{ALM20}.
\end{proof}
\begin{theorem}\label{theorem:RITPS:soundness}
    Upon termination of $\RITPS(\A, \pval, \Init(\A), \emptyset)$, we have:
    \begin{enumerate}
        \item Soundness:
        \begin{itemize}
			\item $\pval \in \RITPS(\A, \pval, \Init(\A), \emptyset)$;
        	\item if $\pval' \in \RITPS(\A, \pval, \Init(\A), \emptyset)$ then
			$\Traces(\valuate{\A}{\pval'}) = \Traces(\valuate{\A}{\pval})$;
        \end{itemize}

        \item Integer completeness: If $\pval'$ is an integer parameter valuation, and $\Traces(\valuate{\A}{\pval'}) = \Traces(\valuate{\A}{\pval})$, then $\pval' \in \RITPS(\A, \pval, \Init(\A), \emptyset)$.
    \end{enumerate}
\end{theorem}
\begin{proof}
	From \cref{proposition:RITPS:soundness,proposition:RITPS:integercompleteness}, and the fact that $\pval$ is an integer valuation.
\end{proof}
\section{Implementation and experiments using \romeo}\label{section:implementation-romeo}

We implemented the rational-valued algorithms (\EF, \AF), integer algorithms (\IEF, \IAF) and mixed algorithms (\RIEF, \RIAF) in the tool \romeo{}~\cite{LRST09};
polyhedra operations (both convex and non-convex) are handled by the Parma Polyhedra Library (PPL)~\cite{BHZ08}.
To illustrate the practical value of our work, we study the application of our algorithms
to the parametric non-preemptive model proposed in~\cite{JLR15} of the classical scheduling problem of~\cite{BFSV04} (\cref{ss:scheduling}),
to the classical Fischer mutual exclusion protocol with the parametric model proposed in~\cite{tripakis-FMSD-01} (\cref{ss:Fischer}) and
to the classical level crossing proposed in~\cite{BV03} (\cref{ss:levelcrossing}).
We used a Core i9/Mac OS computer for our experiments. %

\subsection{Scheduling benchmark}\label{ss:scheduling}

The scheduling example of~\cite{JLR15} consists in three tasks $\tau_1,\tau_2,\tau_3$ scheduled using static priorities ($\tau_1>\tau_2>\tau_3$) in a non-preemptive manner.
Task~$\tau_1$ is periodic with period~$\paramatext$ and a nondeterministic duration in $[10,\parambtext]$, where $\paramatext$ and~$\parambtext$ are parameters.
Task $\tau_2$ only has a minimal activation time of $2 \paramatext$ and has a nondeterministic duration in $[18,28]$ and finally $\tau_3$ is periodic with period $3 \paramatext$ and a nondeterministic duration in $[20,28]$.
First, we ask (using \AF, \IAF{} and \RIAF{} algorithms) for the parameter valuations that ensure that, for all executions, all tasks are executed at least once.
Algorithms $\AF$ and $\IAF$ produce the set of parameter valuations $\paramatext > 24 \land \parambtext \geq 10 \land\paramatext - \parambtext > 14$, while algorithm \RIAF{} produces the set of parameter valuations $\paramatext \geq 25 \land \parambtext \geq 10 \land \paramatext - \parambtext \geq 15$.
The three algorithms take comparable time and memory space--of about $1.0$\,s and~$6.0$\,MiB.

Furthermore, each task is subject to a deadline equal to its period so that it must only have one instance active at all times.  We ask for the parameter valuations that ensure that the system does not reach a deadline violation.\footnote{%
	The result is therefore the complement of the result given by \IEF{} and therefore an over-approximation containing no incorrect integer valuation.
}
Algorithm \EF{} does not terminate.
Algorithm $\IEF$ produces the set of parameter valuations $\paramatext \geq 34 \land \parambtext \geq 10 \land\paramatext - \parambtext\geq 24$ in 13.5\,s, while algorithm $\RIEF$ produces the set of parameter valuations $\paramatext > \frac{562}{17} \land \parambtext \geq 10 \land \paramatext - \parambtext > \frac{392}{17}$ in 14.3\,s.

As illustrated here, the results are indeed a bit more precise but the main improvement is of course the guaranteed density of the result. Also, \RIEF{} is generally slower than \IEF{} and profiling shows that this is due to a decreased efficiency in computing the integer hull: we start each time from the whole symbolic state instead of starting from the successor of an already computed integer hull.
This could be mitigated using a cache for the sets of parameter valuations generated in computing the integer hulls.

\begin{figure}
	\begin{tikzpicture}[PTA, scale=1.2]
 
		\node[location, initial] at (0,0) (l1) {$\loc_1$};
 
		\node[location] at (2,0) (l2) {$\loc_2$};
 
		\node[location] at (4,0) (l3) {$\loc_3$};
 
		\node[location] at (6, 0) (l4) {$\loc_4$};

		\path (l1) edge[below ] node{\begin{tabular}{c @{\ } c}
		& $  \clockx \in [1,4 \parama]$\\
		 & $\clockx\assign0$\\
		\end{tabular}} (l2);

		\path (l2) edge node{\begin{tabular}{c @{\ } c}
	& $  \clockx \in [3\paramb,1]$\\
		 & $\clockx\assign0$\\
		\end{tabular}} (l3);

		\path (l3) edge[below] node{\begin{tabular}{c @{\ } c}
	& $  \clockx \in [\parama,\paramb]$\\
		\end{tabular}} (l4);

	\end{tikzpicture}
	
	\caption{A PTA s.t.\ \IEF{} $\loc_4$ is false and \EF{} and \RIEF{} give $\paramatext \in [\frac{1}{4}, \frac{1}{3}] $ and
$\parambtext \in [\frac{1}{4}, \frac{1}{3}] $ and $\parambtext \geq \paramatext $}
	\label{figure:exempleQuiMarchePasEnInteger}

\end{figure}
\subsection{Fischer's mutual-exclusion protocol}\label{ss:Fischer}

In the following example, we first compare the algorithms when they all have a solution.
Let us therefore consider a well-known benchmark in the literature of real-time verification, introduced in~\cite{lamport-87,abadi-91}.
The system consists of $n$ identical processes which access a critical section by mutual exclusion.
We use the version of this example proposed in~\cite{tripakis-FMSD-01} where $A$ and~$B$ are parameters of the system and the correctness of the protocol depends on their values.

For $n$ processes, mutual exclusion is checked with \romeo\ by the CTL formula $AG \Big(\big( \sum_{i=1}^{n}  \mathit{Critical}[i] \big )\leq 1\Big)$ where $\mathit{Critical}[i]=1$ represents the fact that the process~$i$ is in the critical section ($0$~otherwise).
Recall that $AG(\varphi)$ is equivalent to $\neg EF(\neg \varphi)$.

The property holds iff $A > B$ for \EF{} and \RIEF{} and iff $A \geq B+1$ for \IEF{}.
A comparison of computation times and memory consumption for the mutual exclusion checking  is given in \cref{tab:fisher-comparison-one}.
We are in the case where \EF{} finishes and for this example, \EF{} has comparable performance to \IEF{} but with a more accurate result and is more efficient than \RIEF{} while giving the same result.
The reason is that for this example, the integer hull  used in \RIEF{} adds a computation step without being useful for convergence.

\begin{center}
\begin{table*}
\caption{Results for Fischer protocol with \IEF{}, \EF{}, \RIEF{}}
\centering
\setlength{\tabcolsep}{1.2em}
\hspace{5mm}
\begin{tabular}{llccccc} %
\toprule
          \multicolumn{2}{c} {NB of processes  }   & 4 & 5 & 6 & 7 & 8\\
\midrule
	\IEF{}
& Time (s)  &   $0.3$ & $1.2$  & 5.6   & 23   &  91 \\ %
          & Mem.\ (MB)   & $4.4$ & $10 $ & 36  & 138 & 510    \\
 & Result & \multicolumn{5}{c} {  $A - B \geq 1 $} \\

\midrule
	\EF{}
& Time (s)   & $0.3 $ & 1.1 & 5.2    &  20  & 88.5    \\ %
          & Mem.\ (MB)  &  $3.2$  & $10 $ & 37   & 138 & 516    \\

  \RIEF{}
  & Time (s)    & $0.3 $ & 1.5 & 6.5   & 30    &   118 \\
          & Mem.\ (MB)    & $4.4$ & $13$ & 47   & 177   &  667 \\
  & Result & \multicolumn{5}{c} {  $A - B > 0 $} \\

\bottomrule
\end{tabular}
\label{tab:fisher-comparison-one}
\end{table*}
\end{center}

In order to compare \AF{}, \IAF{} and \RIAF{}, we now consider a Fischer protocol model in which the processes only make a single request to the critical section.
The property to be checked is that for all executions, all processes obtain the critical section in mutual exclusion.
In addition, we change the time interval of the retry from $[A,\infty)$ to $[3 \times A,2]$, which has the effect of producing the constraints
$A \in (0, \frac{2}{3}] $ and
$B \in [0, \frac{2}{3}) $ and
$A - B > 0 $ without integer solution.

\begin{center}
\begin{table*}
\caption{Results for Fischer protocol with \IAF{}, \AF{}, \RIAF{}}
\centering
\setlength{\tabcolsep}{1.2em}
\hspace{5mm}
\begin{tabular}{llccccc} %
\toprule
          \multicolumn{2}{c} {NB of processes  }   & 4 & 5 & 6 & 7 & 8\\
\midrule
	\IAF{}
& Time (s)  &   $0$ & $0.1$  & 0.1   & 0.2   &  0.4 \\ %
          & Mem.\ (MB)   & 0.4  &$0.9$   & 1  & 1.8 & 4.4    \\

& Result & \multicolumn{5}{c} {(empty  set)} \\

\midrule
	\AF{}
& Time (s)   & $0.5 $ & 4.3 & 32    &  238  & 1757    \\ %
          & Mem.\ (MB)  &  $5$  & $29 $ & 181   & 1162 & 6972     \\ %

  \RIAF{}
  & Time (s)    & $0.5 $ & 3 & 21   & 141    &   975 \\
          & Mem.\ (MB)    & $4$ & $19$ & 94   & 503   &  2756 \\
  & Result & \multicolumn{5}{c} {$A \in (0, \frac{2}{3}]$ and $A - B > 0 $} \\

\bottomrule
\end{tabular}
\label{tab:fisher-comparison-two}
\end{table*}
\end{center}

For this example, \IAF{} produces an empty set (no integer solution), while \RIAF{} and \AF{} obtain the same result (\ie{} $A \in (0, \frac{2}{3}]$ and $A - B > 0 $), but \RIAF{} is much more efficient.

\subsection{Level crossing benchmark}\label{ss:levelcrossing}

As for \AF{}, it is easy to show, as in \cref{figure:exempleQuiMarchePasEnInteger} for \EF{}, that there may be no integer valuations of the parameters (\IEF{} returns an empty set) but only non-integer solutions.
In the next example, we will exploit this possibility and compare the algorithms when there is no integer solution, and when there is at least one integer solution.

Let us therefore consider as a third benchmark the classical level crossing example.
We use the version of this example proposed in~\cite{BV03}.
The system is obtained by the parallel composition of $n$~train models synchronized with the gate controller model (actions $\ensuremath{app}$, $\ensuremath{exit}$, with $n$~fixed)
and a gate model (actions $\ensuremath{down}$, $\ensuremath{up}$).

We use a parameter $\paramatext$ for the approach time of trains. This time is either $2\paramatext$ or is in the interval $[\paramatext-1,\paramatext+1]$ with $\paramatext \leq 10$.
We use two parameters $\parambtext$ and $\paramctext$ for the dynamic of opening and closing of the gate.
With EF formulae, we ask for the parameter valuations that ensure that the system does not reach a state such that a train is in front of an open gate.

We consider two cases: first, in case $i$, we choose the  parameters $\parambtext$ and $\paramctext$ such that the parameter valuations to never have a train in front of an open gate are all non-integers.
For this we use a lowering time in $[1, 5 \paramctext]$, a delay to start the opening in $[3 \parambtext,5]$ and an opening time in $[\paramctext,\parambtext]$.
Then, in case $ii$ we modify the lowering time and the delay to start respectively into $[1,\paramctext]$ and $[2\parambtext,5]$   for which there are also integer solutions.

We will consider this example for $n \in \{2, 3, 4\}$~trains.
Adding trains increases the explosion of the state space but does not change the parameter set result.

For the case~$i$, Algorithm \IEF{} produces ``$\emptyset$'' as a result (\ie{} the parameter valuations set is empty) and the result by both \EF{} and \RIEF{} is: $\paramatext \in (\frac{36}{5} , 10] \land
\parambtext \in [\frac{1}{5}, \frac{5}{3}] \land
\paramctext \in [\frac{1}{5}, \frac{2}{3}) \land
 \parambtext - \paramctext \geq 0 \land
 \parambtext - \paramatext + 5 \times \paramctext  < -6$.

The case~$ii$ enlarges the space of solutions allowing some integer valuations.

The result of both \EF{} and \RIEF{} is
$\paramatext \in (8, 10]
\land
\parambtext \in [1, \frac{5}{2}]
\land
\paramctext \in [1, 2)
\land
\parambtext - \paramctext \geq 0
\land
\parambtext -\paramatext + \paramctext  < -6$
while the result of \IEF{} is:
$\paramatext \in [9, 10]
\land
\parambtext = [1, 2]
\land
\paramctext \in [1, 1]
\land
\paramatext - \parambtext  \geq 8 $.

The computation times and memory space used are given in \cref{table:xp}, where ``\TimeOut{}'' denotes no termination after 3~hours.
The results show that \RIEF{} is more precise than \IEF{} and in some cases even provides valuations while \IEF{} produces an empty set---which is a main advantage.\label{newtext:table1}
Furthermore \cref{table:xp} shows that \RIEF{} is most of the time faster than \EF{} and uses less memory.
Although the parameters are bounded in this example, we stopped \EF{} (for 4~trains) after 3~hours of computation without knowing whether the algorithm would eventually terminate or not.
This is not a surprise since \EF{} has no guarantee of termination, whereas \RIEF{}, \IEF{}, \RIAF{} and \IAF{} always terminate for bounded parameters.

\begin{center}
\begin{table*}
\caption{Results for level crossing}
\centering
\setlength{\tabcolsep}{1.1em}
\hspace{5mm}
\begin{tabular}{@{} llccccccc} %
\toprule
 &  & \multicolumn{3}{l}{ i) without integer solution}               &~&\multicolumn{3}{l}{ii) with integer solution}\\
          & Trains    & 2   & 3   & 4  &       & 2    & 3   & 4  \\
\midrule
	\IEF{}
& Time (s)  & $<1$ & $<1$   & $<1$    & & $<1$  & 7   & 948      \\ %
          & Mem.\ (MB) & $<1$ & $<1$ & $<1$   & & $<1$ & 77  & 3735     \\
& Result & \multicolumn{3}{c} {(empty  set)} & &
\multicolumn{3}{c}  {\tiny \begin{tabular}{l}
$a \in [9, 10] $,
$b \in [1, 2] $, 
$c = 1 $ \\
$a-b \geq 8 $
\end{tabular} }   \\
\midrule
	\EF{}
& Time (s)  & 2 & 19     & \TimeOut{}  & & 1 & 11    & \TimeOut{}       \\
          & Mem.\ (MB) & $<1$ & 203   & \TimeOut{}& &  $<1$ & 130   & \TimeOut{}     \\
  \RIEF{}
  & Time (s)  & $<1$ & 4   & 7  &  & 1 & 14   & 2920       \\
          & Mem.\ (MB) & $<1$ & 63  & 168    &  & $<1$ & 269   & 13706     \\
          & Result &   
          \multicolumn{3}{c}
           {\tiny \begin{tabular}{l}
$a \in (\frac{36}{5}, 10] $, \\
$b \in [\frac{1}{5}, \frac{5}{3}] $, \\
$c \in [\frac{1}{5}, \frac{2}{3}) $, \\
$b -c \geq 0$, \\ 
$b -a + 5 \times c  < -6$
\end{tabular}  }
 & &
 \multicolumn{3}{c}   {\tiny \begin{tabular}{l}
$a \in (8, 10]  $,
$b \in [1, \frac{5}{2}]  $, \\
$c \in [1, 2) $,
$ b -c \geq 0 $ \\
$b -a + c  < -6$  
\end{tabular}  } \\
\bottomrule
\end{tabular}
\label{table:xp}
\end{table*}
\end{center}

\section{Implementation and experiments using \imitator}\label{section:implementation-imitator}

Regarding trace preservation, we implemented the rational-valued algorithm (\TPS) and its integer-complete counterpart (\RITPS) in the \imitator{} model checker~\cite{Andre21}.
As in \romeo{}, polyhedra operations are handled by PPL~\cite{BHZ08}.

Note that, from \cref{theorem:RITPS:soundness}, \RITPS{} cannot improve the result of~\TPS{}---but can potentially ensure termination where~\TPS{} would not terminate.

Experiments were conducted on a Dell Precision 5570 with an Intel\textregistered{} Core\texttrademark{} i7-12700H with 32\,GiB memory running Linux Mint 21 Vanessa.
We used \imitator{} \href{https://github.com/imitator-model-checker/imitator/releases}{3.4-alpha2} ``Cheese Durian''.
Binary, models and expected results together with a reproducibility script are available on the long-term archiving platform Zenodo: \url{https://doi.org/10.5281/zenodo.13779414}.

\paragraph{Gaining termination}
First, we apply \TPS{} and \RITPS{} to the example in \cref{figure:RITPS:termine}.
Applying \TPS{} leads to an infinite loop in the algorithm due to the diverging parameter constraints, and therefore (as expected) \TPS{} does not terminate.
In contrast, \RITPS{} terminates in 0.011~second, with the expected result:

\[0 \leq \paramitext{1} \leq \paramitext{2} \leq 5\text{.}\]

While this is not guaranteed by \cref{theorem:RITPS:soundness} (which only ensures integer completeness), the implementation of \RITPS{} in \imitator{} synthesizes for this example not only all integer parameter valuations, but also all rational valuations.
In other words, the synthesised result is the maximal result.

\paragraph{Scalability test}
We also conducted a scalability test to evaluate the overhead induced by the computation of the integer hull when it brings actually no gain at all (\ie{} when comparing the integer hull of the states, instead of the actual states, does not bring any gain in terms of size of the state space, and therefore neither in terms of termination).
We use as benchmark for comparison a model of the Carrier Sense Multiple Access/Collision Detection (CSMA/CD) protocol~\cite{KNSW07} belonging to the \imitator{} benchmarks library~\cite{AMP21}.
This protocol is a network protocol used in Ethernet to manage data transmission by detecting collisions on a shared communication channel.
The idea is as follows: devices first check whether the channel is free, then transmit and, if a collision occurs, they stop, wait for a random time (bounded by an exponential bound), and try again.
The model is made of the parallel composition of three automata, and has three parameters:
$\paramLambda$ (which encodes the length of a message),
$\paramSigma$ (which encodes the propagation time of a message),
and
$\paramTS$ (which encodes the duration of a time slot). %
The models are parametrised by the value of the exponential backoff (``$bc$''); the size of the model grows exponentially with~$bc$.

We fix as reference valuation $\pval$ such that $\pval(\paramLambda) = 808$, $\pval(\paramSigma) = 26$ and $\pval(\paramTS) = 52$.
When calling \TPS{} or \RITPS{} with this parameter valuation, the number of visited states is the same in both algorithms for each value of~$bc$, and the result is exactly the same for any value of~$bc$:

\[ 15 \times \paramTS < \paramLambda < 16 \times \paramTS \land 0 < \paramSigma < \paramTS \text{.}  \]

We give the results in \cref{table:experiments:CSMACD}, where ``\TimeOut{}'' denotes no termination after 5~minutes.
We give from left to right the exponential backoff (\ie{} the scalability factor of this benchmark), the numbers of automata, of clocks, of parameters of the model, the number of locations for all three automata, the number of symbolic states explored (identical for $\TPS$ and $\RITPS$); and the computation time for \TPS{} and \RITPS{} respectively.

\begin{table}
	\caption{Scalability test for \RITPS{} when $\IH$ does not bring state space reduction}
	\begin{tabular}{c c c c r r r r}
	\toprule
	$bc$ & $|\A|$ & $|\Clock|$ & $|\Param|$ & \multicolumn{1}{c}{$|\Loc|$} & \multicolumn{1}{c}{$|\SymbStates|$} & $\TPS$ (s) & $\RITPS$ (s) \\

	\midrule

	1 & 3 & 3 & 3 & $3 \times 8 \times 8$ & \num{511} & \num{0.19} & \num{0.74} \\

	2 & 3 & 3 & 3 & $3 \times 17 \times 17$ & \num{979} & \num{0.38} & \num{1.50} \\

	3 & 3 & 3 & 3 & $3 \times 34 \times 34$ & \num{2249} & \num{0.74} & \num{2.97} \\

	4 & 3 & 3 & 3 & $3 \times 67 \times 67$ & \num{6017} & \num{1.88} & \num{7.70} \\

	5 & 3 & 3 & 3 & $3 \times 132 \times 132$ & \num{16221} & \num{4.95} & \num{20.69} \\

	6 & 3 & 3 & 3 & $3 \times 261 \times 261$ & \num{37177} & \num{12.28} & \num{47.48} \\

	7 & 3 & 3 & 3 & $3 \times 518 \times 518$ & \num{79637} & \num{28.18} & \num{107.35} \\

	8 & 3 & 3 & 3 & $3 \times 1031 \times 1031$ & \num{165105} & \num{72.72} & \TimeOut{} \\ %

	\bottomrule

	\end{tabular}
	\label{table:experiments:CSMACD}
\end{table}

From \cref{table:experiments:CSMACD}, one can infer that the overhead remains acceptable, with a factor almost constant and close to~4 \wrt{}the execution time.
While this is not negligible, we believe this is an acceptable price to pay in order to guarantee termination.

\section{Conclusion}
\label{section:conclusion}

\subsection{Summary}
We introduced here an extrapolation for symbolic states in parametric timed automata (PTAs) that contains not only clocks but also parameters.
We then proposed three algorithms that always terminate for PTAs with bounded (but dense) parameter domains, and output symbolic sets of parameter valuations that define dense sets of parameter valuations that are guaranteed to be correct and containing at least all integer points.

Synthesizing not only the integer points but also the rational-valued points is of utmost importance for the robustness or implementability of the system.
In fact, one can even consider any degree of precision instead of integers (\eg{} a degree of precision of $\frac{1}{10}$) by appropriately resizing the constants of the PTA (\eg{} by multiplying all constants and all parameter bounds by 10).
This makes possible the synthesis of an underapproximated result arbitrarily close to the actual solution.

Our experiments using \romeo{} and \imitator{} show that our algorithms ensure termination for case studies that were otherwise unsolvable using these model checking engines, and incur an acceptable overhead otherwise.

\subsection{Future works}
We proposed a first $\LargestC$-extrapolation for PTAs; this can serve as a basis for further developments, \eg{} using better extrapolation operators such as LU, local-LU or local-diagonal-LU abstractions, following the recent line of works on extrapolations and abstractions over timed automata zones (\eg{} \cite{HSW13,HSW16,BGHSS22}).
The approach we propose is fairly generic and could probably be adapted to more complex properties, expressed in LTL or CTL and their parametric variants.

Furthermore, investigating the possible combination of the integer hull with recent heuristics specifically designed for parametric timed formalisms, such as state merging over symbolic states (\eg{} \cite{AMPP22}) or more complex versions of our parametric timed extrapolation~\cite{AA24} (proposed after~\cite{ALR15}), is on our agenda.

Finally, we use here the integer hull as an underapproximation of the result;
in contrast, we could use an overapproximation using a notion (yet to be defined) of ``external integer hull'', and then combine both hulls to obtain two sets of ``good'' and ``bad'' parameter valuations separated by an arbitrarily small set of unknown valuations.

\section*{Acknowledgment}
\noindent{}
We would like to thank the three anonymous reviewers for useful comments.
We would like to thank an anonymous reviewer for a meaningful remark on a preliminary version of this paper together with the suggestion of the example in \cref{figure:example:reviewerACSD}.
We finally thank Johan Arcile for useful remarks on the manuscript.

	\newcommand{\CCIS}{Communications in Computer and Information Science}
	\newcommand{\ENTCS}{Electronic Notes in Theoretical Computer Science}
	\newcommand{\FAC}{Formal Aspects of Computing}
	\newcommand{\FundInf}{Fundamenta Informaticae}
	\newcommand{\FMSD}{Formal Methods in System Design}
	\newcommand{\IJFCS}{International Journal of Foundations of Computer Science}
	\newcommand{\IJSSE}{International Journal of Secure Software Engineering}
	\newcommand{\IPL}{Information Processing Letters}
	\newcommand{\JAIR}{Journal of Artificial Intelligence Research}
	\newcommand{\JLAP}{Journal of Logic and Algebraic Programming}
	\newcommand{\JLAMP}{Journal of Logical and Algebraic Methods in Programming} %
	\newcommand{\JLC}{Journal of Logic and Computation}
	\newcommand{\LMCS}{Logical Methods in Computer Science}
	\newcommand{\LNCS}{Lecture Notes in Computer Science}
	\newcommand{\RESS}{Reliability Engineering \& System Safety}
	\newcommand{\RTS}{Real-Time Systems}
	\newcommand{\SCP}{Science of Computer Programming}
	\newcommand{\SOSYM}{Software and Systems Modeling ({SoSyM})}
	\newcommand{\STTT}{International Journal on Software Tools for Technology Transfer}
	\newcommand{\TCS}{Theoretical Computer Science}
	\newcommand{\TOPLAS}{{ACM} Transactions on Programming Languages and Systems ({ToPLAS})}
	\newcommand{\ToPNoC}{Transactions on {P}etri Nets and Other Models of Concurrency}
	\newcommand{\TOSEM}{{ACM} Transactions on Software Engineering and Methodology ({ToSEM})}
	\newcommand{\TSE}{{IEEE} Transactions on Software Engineering}
\bibliographystyle{alphaurl}
\bibliography{psynthesis}

\end{document}